\documentclass[11pt,a4paper]{article}
\usepackage{amssymb}
\usepackage{amsmath}
\usepackage{amsfonts}
\usepackage{bbm}
\usepackage{amsthm}
\usepackage{mathrsfs}
\usepackage{hyperref}
\usepackage{marvosym}
\usepackage{color}
\usepackage[margin=2.6cm]{geometry}
\usepackage[all,cmtip]{xy}
\usepackage[utf8]{inputenc}

\usepackage{upgreek}

\definecolor{darkred}{rgb}{0.8,0.1,0.1}
\hypersetup{
     colorlinks=false,         % false: boxed links; true: colored links,false is default
     linkcolor=darkred,
     citecolor=blue,
}

\theoremstyle{plain}
\newtheorem{theo}{Theorem}[section]
\newtheorem{lem}[theo]{Lemma}
\newtheorem{propo}[theo]{Proposition}
\newtheorem{cor}[theo]{Corollary}

\theoremstyle{definition}
\newtheorem{defi}[theo]{Definition}
\newtheorem{ex}[theo]{Example}

\newtheorem{rem}[theo]{Remark}

\numberwithin{equation}{section}

\def\nn{\nonumber}

\def\bbR{\mathbb{R}}
\def\bbC{\mathbb{C}}

\def\bbZ{\mathbb{Z}}
\def\bbT{\mathbb{T}}

\def\Hom{\mathrm{Hom}}

\def\id{\mathrm{id}}
\def\supp{\mathrm{supp}}

\def\dd{\mathrm{d}}

\def\sc{\mathrm{sc}}

\def\1{\mathbbm{1}}

\def\op{\mathrm{op}}

\newcommand{\ip}[2]{\big\langle #1,#2 \big\rangle}
\newcommand{\ips}[2]{\langle #1,#2\rangle}

\def\Conf{\mathfrak{C}}
\def\Obs{\mathfrak{O}}

\def\sdConf{\mathfrak{sdC}}
\def\sdObs{\mathfrak{sdO}}

\def\Man{\mathsf{Man}}

\def\Ab{\mathsf{Ab}}

\def\Loc{\mathsf{Loc}}
\def\PSAb{\mathsf{PSAb}}
\def\CAlg{C^\ast\mathsf{Alg}}

\def\TT{\mathfrak{T}}
\def\obs{\mathcal{O}}

\def\OmegaZ{\Omega_\bbZ}
\def\free{\mathrm{free}}
\def\tor{\mathrm{tor}}
\def\dR{\mathrm{dR}}

\def\H{\mathrm{H}}
\def\dH{\hat{\H}}
\def\Ht{\H_\tor}
\def\Hf{\H_\free}

\def\Pair{\mathsf{Pair}}

\def\cu{\mathrm{curv}}
\def\dcu{\ast\,\cu}
\def\ch{\mathrm{char}}
\def\de{\delta}
\def\del{\partial}
\def\cdel{\updelta}

\def\colim{\mathrm{colim}}

\def\solve{\mathrm{solve}}

\def\sk{\vspace{2mm}}

\title{%
Abelian duality on globally hyperbolic spacetimes
}

\author{%
Christian Becker$^{1,a}$, Marco Benini$^{1,2,b}$, Alexander Schenkel$^{2,c}$ and Richard J.\ Szabo$^{2,d}$ \vspace{4mm}\\
{\small $^1$ Institut f\"ur Mathematik, Universit\"at Potsdam,}\\
{\small Karl-Liebknecht-Str.~24-25, 14476 Potsdam, Germany.}\vspace{2mm}\\
{\small $^2$ Department of Mathematics, Heriot-Watt University,}\\
{\small Colin Maclaurin Building, Riccarton, Edinburgh EH14 4AS, United Kingdom.}\vspace{0.2mm}\\
{\small  Maxwell Institute for Mathematical Sciences, Edinburgh, United Kingdom.}\vspace{0.2mm}\\
{\small  Higgs Centre for Theoretical Physics, Edinburgh, United Kingdom.}\vspace{4mm}\\
{\footnotesize \texttt{Email:} $^a$\texttt{becker@math.uni-potsdam.de},
$^b$\texttt{mbenini87@gmail.com}, $^c$\texttt{as880@hw.ac.uk}, $^d$\texttt{R.J.Szabo@hw.ac.uk}}
 }

\date{May 2016}

\begin{document}

\maketitle

\begin{abstract}
\noindent
We study generalized electric/magnetic duality in Abelian gauge theory by combining techniques from locally covariant quantum field theory and Cheeger-Simons differential cohomology on the category of globally hyperbolic Lorentzian manifolds. Our approach generalizes previous treatments using the Hamiltonian formalism in a manifestly covariant way and without the assumption of compact Cauchy surfaces. We construct semi-classical configuration spaces and corresponding presymplectic Abelian groups of observables, which are quantized by the CCR-functor to the category of $C^*$-algebras. We demonstrate explicitly how duality is implemented as a natural isomorphism between quantum field theories. We apply this formalism to develop a fully covariant quantum theory of self-dual fields.
\end{abstract}

\paragraph*{Report no.:} EMPG--15--13

\newpage

{\baselineskip=12pt
\tableofcontents
}

\bigskip

\section{\label{secIntro}Introduction and summary}
Dualities in string theory have served as a rich source of
conjectural relations between seemingly disparate situations in
mathematics and physics, particularly in some approaches to quantum
field theory. Heuristically, a `duality' is an equivalence between two
descriptions of the same quantum theory in different classical terms,
and it typically involves an interchange of classical and quantum
data. The prototypical example is electric/magnetic duality of Maxwell
theory on a four-manifold $M$: Magnetic flux is discretized at the
classical level by virtue of the fact that it originates as the
curvature of a line bundle on $M$, whereas electric flux
discretization is a quantum effect arising via Dirac charge
quantization. The example of electric/magnetic duality in Maxwell
theory has a generalization to any spacetime dimensionality, of relevance to 
the study of fluxes in string theory, which we may collectively refer 
to as `Abelian duality'. The configuration spaces of these (generalized) Abelian gauge theories 
are mathematically modeled by suitable (generalized) differential cohomology groups, 
see e.g.~\cite{Freed:2000ta,Szabo:2012hc} for reviews.
\sk

In this paper we will describe a new perspective on Abelian duality by
combining methods from Cheeger-Simons differential cohomology and
locally covariant quantum field theory; this connection between
Abelian gauge theory and differential cohomology was originally pursued 
by~\cite{Becker:2014tla}. The quantization of Abelian gauge theories was described 
from a Hamiltonian perspective by~\cite{Freed:2006ya,Freed:2006yc}, where the 
representation theory of Heisenberg groups was used to define the quantum Hilbert 
space of an Abelian gauge theory in a manifestly duality invariant way. 
In the present work we shall instead build the semi-classical configuration space for dual gauge field 
configurations in a fully covariant fashion, which agrees with that proposed 
by~\cite{Freed:2006ya,Freed:2006yc} upon fixing a Cauchy surface $\Sigma$ in 
a globally hyperbolic spacetime $M$, but which is manifestly independent of the 
choice of $\Sigma$. Following the usual ideas of algebraic quantum field theory, 
we construct not a quantum Hilbert space of states but rather a $C^*$-algebra of 
quantum observables; the requisite natural presymplectic structure also agrees with 
that of~\cite{Freed:2006ya,Freed:2006yc} upon fixing a Cauchy surface $\Sigma$, 
but is again independent of the choice of $\Sigma$. Our approach thereby lends a 
new perspective on the phenomenon of Abelian duality, and it enables a rigorous 
(functorial) definition of quantum duality as a natural isomorphism between quantum 
field theory functors. An alternative rigorous perspective on Abelian duality has been recently proposed
by~\cite{Elliott:2014haa} using the factorization algebra approach to (Euclidean) quantum field theory.
We do not yet understand how to describe the full duality groups, 
i.e.\ the analogues of the $SL(2,\bbZ)$ S-duality group of Maxwell theory, as this
in principle requires a detailed understanding of the automorphism groups of our 
quantum field theory functors~\cite{Fewster:2012yc}, which is beyond the scope of the present paper.
\sk

Our approach also gives a novel and elegant formulation of the quantum theory
of self-dual fields, which is an important ingredient in the
formulation of string theory and supergravity: In two dimensions the self-dual gauge
field is a worldsheet periodic chiral scalar field in heterotic string theory whose quantum Hilbert space
carries representations of the usual (affine) Heisenberg algebra; in six
dimensions the self-dual field is an Abelian gerbe connection which lives on the worldvolume 
of M5-branes and NS5-branes, and in the evasive superconformal
$(2,0)$~theory whose quantum Hilbert space should similarly carry
irreducible representations of the corresponding Heisenberg group; in ten dimensions 
the self-dual field is the Ramond-Ramond four-form potential of Type~IIB supergravity.
The two generic issues
associated with the formulation of the self-dual field theory are: (a) The lack of covariant local
Lagrangian formulation of the theory (without certain choices,
cf.~\cite{Belov:2006jd}); and: (b) The reconciliation of the
self-duality equation with Dirac quantization requires the
simultaneous discretization of both electric and magnetic fluxes in
the same semi-classical theory. Our quantization of Abelian gauge theories
at the level of algebras of quantum observables eludes both of these
problems. In particular, the noncommutativity of torsion fluxes
observed by~\cite{Freed:2006ya,Freed:2006yc} is also straightforwardly
evident in our approach. As in~\cite{Freed:2006ya,Freed:2006yc}, our quantization
 of the self-dual field does not follow from the approach developed in the rest of this paper. 
Other Abelian self-dual gauge theories can be analyzed starting 
from generalized differential cohomology theories fulfilling a suitable self-duality property, 
e.g.\ differential K-theory, see~\cite{Freed:2006ya,Freed:2006yc} for the Hamiltonian point of view. 
An approach closer to the one pursued in the present paper is possible also in these cases 
provided one has suitable control on the properties of the relevant generalized differential cohomology theory. 
\sk

In addition to being cast in a manifestly covariant framework, another improvement on 
the development of~\cite{Freed:2006ya,Freed:2006yc} is that our approach does not 
require the spacetime to admit compact Cauchy surfaces. Our main technical achievement 
is the development of a suitable theory of Cheeger-Simons differential characters  with compact support
and Pontryagin duality, in a manner which does not destroy the Abelian 
duality. As the mathematical details of this theory are somewhat involved and of independent 
interest, they have been delegated to a companion paper~\cite{BBSS} to which we frequently 
refer. The present paper focuses instead on the aspects of interest in physics.
\sk

The outline of the remainder of this paper is as follows. In Section~\ref{secConf} we 
introduce and analyze the semi-classical configuration spaces of dual gauge fields in the 
language of differential cohomology; our main result is the identification of this space 
with the space of solutions of a well-posed Cauchy problem which agrees with the description 
of~\cite{Freed:2006ya,Freed:2006yc}, but in a manifestly covariant fashion and without 
the assumption of compactness of Cauchy surfaces. In Section~\ref{secSCConf} we 
analogously study a suitable space of dual gauge field configurations of spacelike compact 
support, and show in Section~\ref{secObs} that it is isomorphic to a suitable Abelian group 
of observables defined in the spirit of smooth Pontryagin duality as in~\cite{Becker:2014tla}. 
In Section~\ref{secQuantization} we consider the quantization of the semi-classical gauge theories 
and the extent to which they satisfy the axioms of locally covariant quantum field 
theory~\cite{Brunetti:2001dx}; we show that, just as in~\cite{Becker:2014tla}, our quantum 
field theory functors satisfy the causality and time-slice axioms but violate the locality 
axiom.\footnote{The violation of locality is due to topological properties of the spacetime $M$ and 
owes to the fact that differential cohomology constructs the pertinent configuration spaces 
as gauge orbit spaces. As a matter of fact, all approaches to gauge theory in the context of 
general local covariance \cite{Brunetti:2001dx} exhibit at least some remnant of the failure of locality, 
see \cite{Becker:2014tla, Benini:2013ita, Benini:2013tra, Dappiaggi:2011zs, 
Fewster:2014hba, Dappiaggi:2011cj, Sanders:2012sf}. 
There are indications that the tension between locality and gauge theory can be solved 
by means of homotopical techniques (in the context of model categories), 
see~\cite{Benini:2015hta} for the first steps towards this goal.} 
In Section~\ref{secAbeliandualities} we show that dualities extend to the 
quantum field theories thus defined. In Section~\ref{secSelfdual} we apply our formalism to give a 
proper covariant formulation of the quantum field theory of a self-dual field. An appendix at the 
end of the paper provides some technical details of constructions which are used in the main text.

%%%%%%%%%%%%%%%%%%%%%%%%%%%%%%%%%%%%%%%%%%%%%%%%
%%%%%%%%%%%%%%%%%%%%%%%%%%%%%%%%%%%%%%%%%%%%%%%%

\section{\label{secConf}Dual gauge fields}
In this section we describe and analyze the configuration spaces of
the (higher) gauge theories that will be of interest in this paper.
Their main physical feature is a discretization of both electric and magnetic 
fluxes, which is motivated by Dirac charge quantization.
To simplify notation, we normalize both 
electric and magnetic fluxes so that they are quantized in the same
integer lattice $\bbZ\subset \bbR$.
Because Dirac charge quantization arises as a quantum effect
(i.e.~it depends on Planck's constant $\hbar$, which in our
conventions is equal to $1$), 
we shall use the attribution ``semi-classical'' for the gauge
field configurations introduced below. In this paper all manifolds are implicitly assumed 
to be smooth, connected, oriented and of finite type, i.e.\ they admit a finite good cover.

\subsection{Semi-classical configuration space}
Let $M$ be a manifold. The integer cohomology group $\H^k(M;\bbZ)$ of degree $k$
is an Abelian group which has a (non-canonical) splitting
$\H^k(M;\bbZ)\simeq \Hf^k(M;\bbZ)\oplus \Ht^k(M;\bbZ)$ into free and torsion
subgroups, respectively. Let $\Omega_\bbZ^k(M)\subset \Omega^k(M)$
denote the closed differential $k$-forms on $M$ with integer periods. 
Below we recall the definition of Cheeger-Simons differential characters~\cite{CS}. 
\begin{defi}\label{defDiffChar}
A degree $k$ Cheeger-Simons differential character on a manifold $M$
is a group homomorphism $h : Z_{k-1}(M)\to \bbT$ 
from the group $Z_{k-1}(M)$ of $k-1$-cycles on $M$ to the circle group $\bbT := \bbR/\bbZ$
for which there exists a differential form $\omega_h\in\Omega^k(M)$ such that
\begin{flalign}\label{eqn:curvproperty}
h(\del \gamma) = \int_\gamma\, \omega_h \mod \bbZ~,\qquad \forall \gamma \in C_k(M)~,
\end{flalign}
where $\del \gamma$ denotes the boundary of the $k$-chain $\gamma$. 
The Abelian group of Cheeger-Simons differential characters is denoted by $\dH^k(M;\bbZ)$. 
\end{defi}
For a modern perspective on differential cohomology which includes the Cheeger-Simons model
see \cite{SS,BB}. We use the degree conventions of \cite{BB} in which the curvature 
of a differential character in $\dH^k(M;\bbZ)$ is a $k$-form.
The assignment of $\dH^k(M;\bbZ)$ to each manifold $M$ 
is a contravariant functor
\begin{flalign}\label{eqdiffcohofunctor}
\dH^k( - ;\bbZ) : \mathsf{Man}^{\op}\longrightarrow \Ab
\end{flalign}
from the category $\Man$ of manifolds to the category $\Ab$ of Abelian groups. For 
notational convenience, we simply denote by $f^\ast$ the group homomorphism
$\dH^k(f;\bbZ) : \dH^k( M^\prime ;\bbZ) \to \dH^k( M ;\bbZ) $ for any
smooth map $f: M\to M^\prime$.
The functor (\ref{eqdiffcohofunctor}) comes together with four natural transformations which are given by
the curvature map $\cu : \dH^k( - ;\bbZ) \Rightarrow \Omega^k_{\bbZ}(-)$, 
the characteristic class map $\ch : \dH^k( - ;\bbZ) \Rightarrow \H^k(-;\bbZ)$, 
the inclusion of topologically trivial fields $\iota : \Omega^{k-1}(-)/\OmegaZ^{k-1}(-) \Rightarrow \dH^k( - ;\bbZ) $ 
and the inclusion of flat fields $\kappa: \H^{k-1}(-;\bbT)\Rightarrow
\dH^k( - ;\bbZ) $, where $\bbT=\bbR / \bbZ$ is the circle group.
The (functorial) diagram of Abelian groups
\begin{flalign}\label{eqDiffCohoDiagram}
\xymatrix{
		&	0 \ar[d]	&	0 \ar[d]	&	0\ar[d]																\\
0 \ar[r]	&	\frac{\H^{k-1}(M;\bbR)}{\Hf^{k-1}(M;\bbZ)} \ar[r]^-{\tilde \kappa} \ar[d]
					&	\frac{\Omega^{k-1}(M)}{\OmegaZ^{k-1}(M)} \ar[r]^-\dd \ar[d]_-\iota
								&	\dd\Omega^{k-1}(M) \ar[r] \ar[d]^-{\subseteq}					&	0	\\
0 \ar[r]	&	\H^{k-1}(M;\bbT) \ar[r]^-\kappa \ar[d]
					&	\dH^k(M;\bbZ) \ar[r]^-\cu \ar[d]_-\ch
								&	\OmegaZ^k(M) \ar[r] \ar[d]^-{[\,\cdot\,]}					&	0	\\
0 \ar[r]	&	\Ht^k(M;\bbZ) \ar[r] \ar[d]
					&	\H^k(M;\bbZ) \ar[r] \ar[d]
								&	\Hf^k(M;\bbZ) \ar[r] \ar[d] 									&	0	\\
		&	0		&	0		&	0
}
\end{flalign}
is a commutative diagram whose rows and columns are short exact sequences.
\vspace{2mm}

In the remainder of this paper we shall take $M$ to be a time-oriented $m$-dimensional globally hyperbolic
Lorentzian manifold, which we regard as `spacetime'; for a thorough
discussion of Lorentzian geometry including global hyperbolicity see
e.g.\ \cite{BEE,ONeill}, while a brief overview can be found in e.g.~\cite[Section 1.3]{Bar:2007zz}. 
The semi-classical configuration space $\Conf^k(M;\bbZ)$ of interest to us is obtained as the pullback 
\begin{flalign}\label{eqConf}
\xymatrix{
\Conf^k(M;\bbZ) \ar@{-->}[d] \ar@{-->}[r]	&	\dH^{m-k}(M;\bbZ) \ar[d]^-\dcu		\\
\dH^k(M;\bbZ) \ar[r]_-\cu					&	\Omega^k(M)
}
\end{flalign}
By definition, any element $(h,\tilde h) \in \Conf^k(M;\bbZ)\subseteq \dH^k(M;\bbZ) \times \dH^{m-k}(M;\bbZ) $ 
has the property that the curvature of $h$ is the Hodge dual of the
curvature of $\tilde h$, i.e.\ $\cu\, h = \dcu\, \tilde h$. 
We may interpret this condition as being responsible for the
quantization of electric fluxes: the de~Rham cohomology class
of the Hodge dual curvature $\dcu \,h$ is also an element in $\Hf^{m-k}(M;\bbZ)$
and hence electric fluxes are quantized in the same lattice
$\bbZ\subset \bbR$ as magnetic fluxes.
In a similar fashion, we introduce the semi-classical topologically trivial fields $\TT^k(M;\bbZ)$ as the pullback 
\begin{flalign}\label{eqTTConf}
\xymatrix{
\TT^k(M;\bbZ) \ar@{-->}[d] \ar@{-->}[r]		&	\frac{\Omega^{m-k-1}(M)}{\OmegaZ^{m-k-1}(M)}
														\ar[d]^-{\ast\,\dd}								\\
\frac{\Omega^{k-1}(M)}{\OmegaZ^{k-1}(M)}
	\ar[r]_-\dd									&	\Omega^k(M)
}
\end{flalign}
\vspace{1mm}

To simplify notation we will adopt the following useful convention: 
For any graded Abelian group $A^\sharp = \bigoplus_{k\in\bbZ} \, A^k$, we introduce 
\begin{flalign}
A^{p,q} := A^p \times A^q~. 
\end{flalign}
Using \eqref{eqDiffCohoDiagram} we introduce
a new commutative diagram of Abelian groups with exact rows and columns, 
whose central object is the semi-classical configuration space $\Conf^k(M;\bbZ)$. 
\begin{theo}\label{thmConfDiagram}
Consider the two group homomorphisms 
\begin{subequations}
\begin{flalign}
\cu_1: \Conf^k(M;\bbZ) \longrightarrow \OmegaZ^k \cap \ast
\OmegaZ^{m-k}(M)~, \qquad
(h,\tilde h) \longmapsto \cu\, h = \dcu \,\tilde h
\end{flalign}
and 
\begin{flalign}
\dd_1: \TT^k(M;\bbZ) \longrightarrow \dd \Omega^{k-1} \cap \ast \dd
\Omega^{m-k-1}(M)~, \qquad
([A],[\tilde A]) \longmapsto \dd A = \ast \dd \tilde A~.
\end{flalign}
\end{subequations}
Then the diagram of Abelian groups 
\begin{flalign}\label{eqConfDiagram}
\xymatrix{
		&	0 \ar[d]		&	0 \ar[d]		&	0\ar[d]														\\
0 \ar[r]	&	\frac{\H^{k-1,m-k-1}(M;\bbR)}{\Hf^{k-1,m-k-1}(M;\bbZ)} 
				\ar[r]^-{\tilde \kappa \times \tilde \kappa} \ar[d]
						&	\TT^k(M;\bbZ) \ar[r]^-{\dd_1} \ar[d]_-{\iota \times \iota}
										&	\dd\Omega^{k-1} \cap \ast\,\dd\Omega^{m-k-1}(M) 
												\ar[r] \ar[d]^-\subseteq							&	0	\\
0 \ar[r]	&	\H^{k-1,m-k-1}(M;\bbT) \ar[r]^-{\kappa \times \kappa} \ar[d]
						&	\Conf^k(M;\bbZ) \ar[r]^-{\cu_1} \ar[d]_-{\ch \times \ch}
										&	\OmegaZ^k \cap \ast\,\OmegaZ^{m-k}(M)
												\ar[r] \ar[d]^-{([\,\cdot\,],[\ast^{-1}\,\cdot\,])}	&	0	\\
0 \ar[r]	&	\Ht^{k,m-k}(M;\bbZ) \ar[r] \ar[d]
						&	\H^{k,m-k}(M;\bbZ) \ar[r] \ar[d]
										&	\Hf^{k,m-k}(M;\bbZ) \ar[r] \ar[d] 					&	0	\\
		&	0			&	0			&	0
}
\end{flalign}
is a commutative diagram whose rows and columns are short exact sequences.
\end{theo}
\begin{proof}
Commutativity of this diagram follows by construction. Hence we focus on proving that the rows and columns are exact. 
The bottom row and the left column are exact because they are Cartesian products of exact sequences. 
Injectivity of $\iota \times \iota$, $\kappa \times \kappa$ 
and $\tilde \kappa \times \tilde \kappa$ is immediate by \eqref{eqDiffCohoDiagram}. 
\sk

Let us now show that $\dd_1$ and $\cu_1$ are surjective. 
Given $\dd A = \ast\, \dd \tilde A$ for $A \in \Omega^{k-1}(M)$ 
and $\tilde A \in \Omega^{m-k-1}(M)$, we note that $([A], [\tilde A])$ is an element of $\TT^k(M; \bbZ)$ 
and $\dd_1 ([A], [\tilde A]) = \dd A =\ast\, \dd \tilde A$, thus showing that $\dd_1$ is surjective. 
A similar argument applies to $\cu_1$ using surjectivity of $\cu: \dH^p(M;\bbZ) \to \OmegaZ^p(M)$ 
for $p = k$ and for $p = m-k$. 
\sk

To show that $([\,\cdot\,],[\ast^{-1}\,\cdot\,])$ is also surjective,
let us take any $(z,\tilde z) \in \Hf^{k,m-k}(M;\bbZ) \subseteq \H^{k,m-k}(M;\bbR)$
and recall that by de Rham's theorem it can be presented as $(z,\tilde z) = ([\omega],[\tilde \omega])$, 
for some $\omega \in \OmegaZ^k(M)$ and $\tilde \omega \in \OmegaZ^{m-k}(M)$. 
Let $\de = (-1)^{m\, (k-1)}\, \ast\dd\,\ast: \Omega^k(M) \to \Omega^{k-1}(M)$ denote the codifferential. 
We solve the equations $[\de \theta] = [\omega] \in \H^k_\dR(M)$ 
and $[\de \tilde \theta\, ]=[\tilde \omega] \in \H^{m-k}_\dR(M)$ 
for $\theta \in \Omega^{k+1}(M)$ and $\tilde \theta \in \Omega^{m-k+1}(M)$.\footnote{To show that 
a solution exists, let us introduce the d'Alembert operator  
$\Box = \de \, \dd + \dd \, \de$ and consider its retarded/advanced 
Green's operators $G^\pm$, cf.\ \cite{Bar, Bar:2007zz}. 
Let us also consider a partition of unity $\{\chi_+,\chi_-\}$ on $M$ such that 
$\chi_\pm$ has past/future compact support, see \cite{Bar} for a definition of these support systems. 
Then $\theta = G (\dd \chi_+ \wedge \omega)$ is a solution, where $G = G^+ - G^-$ is the causal propagator. 
In fact $\de \theta = \de \, \dd (G^+ (\chi_+ \, \omega) + G^- (\chi_-
\, \omega)) 
= \omega - \dd \, \de (G^+ (\chi_+ \, \omega) + G^- (\chi_- \, \omega))$. 
A similar argument applies to $\tilde \theta$.} 
Introducing $F = \de \theta + \ast\, \de \tilde \theta$, we find $[F] = [\de \theta] = [\omega] \in \H^k_\dR(M)$ 
and $[\ast^{-1} F] = [\de \tilde \theta\, ] = [\tilde \omega] \in \H^{m-k}_\dR(M)$; 
in particular, both $F$ and $\ast^{-1} F$ have integral periods since so do $\omega$ and $\tilde \omega$. 
We conclude that $F \in \OmegaZ^k \cap \ast\, \OmegaZ^{m-k}(M)$. 
\sk

Surjectivity of $\ch \times \ch$ follows from what we have already shown above 
and by using a diagram chasing argument. 
Take any $(x,\tilde x) \in \H^{k,m-k}(M;\bbZ)$. Mapping to the corresponding free group 
and recalling that both $([\,\cdot\,],[\ast^{-1}\,\cdot\,])$ and $\cu_1$ are epimorphisms, 
we find $(h,\tilde h) \in \Conf^k(M;\bbZ)$ whose image along $([\,\cdot\,],[\ast^{-1}\,\cdot\,]) \circ \cu_1$ 
matches the image of $(x,\tilde x)$ in $\Hf^{k,m-k}(M;\bbZ)$.
By exactness of the bottom row, 
$(\ch\, h,\ch\, \tilde h)$ differs from $(x, \tilde x)$ by an element
$(t,\tilde t\, )$ 
of the torsion subgroup $\Ht^{k,m-k}(M;\bbZ)$, i.e.\ $(x, \tilde x) =
(\ch\, h + t,\ch\, \tilde h + \tilde t\, )$. 
Exactness of the left column allows us to find a preimage $(u,\tilde
u) \in \H^{k-1,m-k-1}(M;\bbT)$ for $(t,\tilde t\, )$. 
Commutativity of the diagram then implies that $(h + \kappa\, u,\tilde h + \kappa\, \tilde u)$ 
is a preimage of $(x,\tilde x)$ via $\ch \times \ch$. 
\sk

We still have to check that the first two rows and the last two columns are exact at their middle objects. 
This is a straightforward consequence of the exactness of the corresponding rows and columns in \eqref{eqDiffCohoDiagram}. 
\end{proof}

\begin{rem}
To better motivate the semi-classical configuration space $\Conf^k(M;\bbZ)$ 
we establish below its relation with Maxwell theory. For this purpose we consider the case $m=4$ and $k=2$. 
The usual Maxwell equations (without external sources) for the Faraday tensor $F \in \Omega^2(M)$ 
are $\dd F = 0$ and $\dd \ast F = 0$. These equations are invariant under electric-magnetic duality, 
i.e.\ under the exchange of $F$ and $\ast F$. The standard approach to gauge theory 
consists in the replacement of $F$ with the curvature of (the isomorphism class of) a circle bundle with connection
(equivalently, a differential cohomology class in degree 2). 
In this framework, however, $\ast F$ does not have any geometric interpretation, 
hence the original electric-magnetic duality of Maxwell theory is lost passing to gauge theory. 
Nevertheless, one can present Maxwell equations in an equivalent way, 
which is however better suited for a gauge theoretic extension preserving electric-magnetic duality: 
\begin{flalign}
F = \ast \tilde F~,~~~~~ \dd F = 0~,~~~~~ \dd \tilde F = 0~.
\end{flalign}
Interpreting both $F$ and $\tilde F$ as the curvatures of circle bundles with connections, 
the semi-classical configuration space $\Conf^2(M;\bbZ)$ is obtained 
and the original electric-magnetic duality of Maxwell theory is lifted to $\Conf^2(M;\bbZ)$, 
see Section \ref{secAbeliandualities} for the situation in arbitrary spacetime dimension and degree. 
Notice that the semi-classical configuration space has the same local ``degrees of freedom'' as Maxwell theory. 
In fact, on a contractible spacetime $\Conf^2(M;\bbZ)$ reduces to
the top-right corner in diagram \eqref{eqConfDiagram}. 
Since exact and closed forms are the same on a contractible manifold, Maxwell theory is recovered. 
In conclusion, the semi-classical configuration space $\Conf^2(M;\bbZ)$ 
is a gauge theoretic extension of Maxwell theory that carries the same local information, 
however preserving electric-magnetic duality by matching  
the relevant topological (as opposed to local) data in a suitable way. 
As a by-product, any configuration $(h,\tilde h) \in \Conf^2(M;\bbZ)$ 
realizes the discretization of magnetic and electric fluxes, 
which arise as the characteristic classes $\ch \,h, \ch\, \tilde h \in \H^2(M;\bbZ)$. 
This argument can be made general for higher gauge theories in arbitrary spacetime dimension. 
\end{rem}

\begin{rem}
The semi-classical configuration space is a contravariant functor 
\begin{flalign}
\Conf^k(-;\bbZ): \Loc_m^\op \longrightarrow \Ab
\end{flalign}
from the category $\Loc_m$ of time-oriented $m$-dimensional globally hyperbolic Lorentzian manifolds 
with causal embeddings\footnote{A causal embedding $f: M \to M^\prime$ 
between time-oriented $m$-dimensional globally hyperbolic Lorentzian manifolds 
is an orientation and time-orientation preserving isometric embedding, 
whose image is open and causally compatible, 
i.e.\ $J^\pm_{M^\prime}(f(p)) \cap f(M) = f(J^\pm_M(p))$ for all $p
\in M$; here $J^\pm_M(p)$ denotes the causal future/past of $p\in M$
consisting of all points of $M$ which can be reached by a
future/past-directed smooth causal curve stemming from $p$, see~\cite{Bar:2007zz}.} as morphisms 
to the category $\Ab$ of Abelian groups. For notational convenience, we simply denote by $f^\ast$ 
the group homomorphism $\Conf^k(f;\bbZ): \Conf^k(M^\prime;\bbZ) \to
\Conf^k(M;\bbZ)$ associated with a morphism
$f: M \to M^\prime$ in $\Loc_m$. 
\end{rem}

%%%%%%%%%%%%%%%%%%%%%%%%%%%%%%%%%%%%%%%%%%%%%%%%
%%%%%%%%%%%%%%%%%%%%%%%%%%%%%%%%%%%%%%%%%%%%%%%%

\subsection{Cauchy problem}\label{secCauchyPbl}
We will now show that the semi-classical configuration space
$\Conf^k(M;\bbZ)$ is the space of solutions of a well-posed Cauchy problem. 
Let us start by recalling a well-known result for the Cauchy problem 
of the Faraday tensor, see e.g.\ \cite{Dappiaggi:2011zs, Fewster:2014hba}
and also \cite[Chapter 3, Corollary 5]{BaerFredenhagen} for details on 
how to treat initial data of not necessarily compact
support. For the related Cauchy problem of the gauge potential see \cite{Sanders:2012sf}. 
Throughout this paper $\Sigma$ will denote a smooth 
spacelike Cauchy surface of $M$ with embedding $\iota_\Sigma: \Sigma \to M$ into~$M$.
\begin{theo}\label{thmCauchyPblF}
For each $(B,\tilde B) \in \Omega^{k,m-k}_\dd(\Sigma)$ (where the subscript ${}_\dd$ denotes closed forms),
there exists a unique solution $F \in \Omega^k(M)$ to the initial value problem 
\begin{subequations}\label{eqCauchyPblF}
\begin{align}
\dd F 			&	= 0~, \qquad	\iota_\Sigma^\ast F	= B~,			\\[4pt]
\dd \ast^{-1} F	&	= 0~, \qquad	\iota_\Sigma^\ast \ast^{-1} F	= \tilde B~,
\end{align}
\end{subequations}
whose support is contained
in the causal future and past of the support of the initial data, i.e.\
$\supp\, F \subseteq J(\supp\, B \cup \supp\, \tilde B)$.
\end{theo}

We consider also the similar well-posed initial value problem for $\tilde F \in \Omega^{m-k}(M)$ given by
\begin{subequations}\label{eqCauchyPblDualF}
\begin{align}
\dd \tilde F 		&	= 0~,	\qquad	\iota_\Sigma^\ast \tilde F	= \tilde B~,		\\[4pt]
\dd \ast \tilde F	&	= 0~, \qquad	\iota_\Sigma^\ast \ast \tilde F	= B~,
\end{align}
\end{subequations}
where the initial data are also specified by $(B,\tilde B) \in \Omega^{k,m-k}_\dd(\Sigma)$.
\sk

Given now any initial data $(B,\tilde B) \in \Omega^{k,m-k}_\dd(\Sigma)$,
let us consider the corresponding unique solutions $F$ and $\tilde F$ of 
the Cauchy problems \eqref{eqCauchyPblF} and \eqref{eqCauchyPblDualF}.
This implies that $F - \ast \tilde F$ solves the Cauchy problem \eqref{eqCauchyPblF} with vanishing initial data, 
and therefore $F = \ast \tilde F$. 
We further show that, given initial data $(B,\tilde B) \in \Omega^{k,m-k}_\bbZ(\Sigma)$ with integral periods,
the corresponding solution $F$ of the Cauchy problem \eqref{eqCauchyPblF}
is such that both $F$ and $\ast^{-1} F$ have integral periods. 
For this, using the results of Lemma \ref{lemChainHomotopy} (i)
we can express each $k$-cycle $\gamma \in Z_k(M)$ as 
$\gamma = \iota_{\Sigma\,\ast} \, \pi_{\Sigma\,\ast} \gamma + \del
h_\Sigma \gamma$, and hence 
\begin{flalign}\label{eqZCauchyPblF}
\int_\gamma \, F = \int_{\pi_{\Sigma\,\ast} \gamma} \,
\iota_\Sigma^\ast F + \int_{h_\Sigma \gamma} \, \dd F 
= \int_{\pi_{\Sigma\,\ast} \gamma} \, B \ \in \ \bbZ~.
\end{flalign}
Similarly, for each $m{-}k$-cycle $\tilde \gamma \in Z_{m-k}(M)$ we have 
\begin{flalign}\label{eqZCauchyPblDualF}
\int_{\tilde \gamma}\, \ast^{-1} F 
= \int_{\pi_{\Sigma\,\ast} \tilde \gamma} \, \iota_\Sigma^\ast\, \ast^{-1} F 
+ \int_{h_\Sigma \tilde \gamma} \, \dd \ast^{-1} F 
= \int_{\pi_{\Sigma\,\ast} \tilde \gamma}\, \tilde B \ \in \ \bbZ~. 
\end{flalign}
Conversely, given $F \in \Omega^k_\bbZ \cap \ast \Omega^{m-k}_\bbZ(M)$
we have $\iota_\Sigma^\ast F \in \Omega^k_\bbZ(\Sigma)$ 
and $\iota_\Sigma^\ast \ast^{-1} F \in \Omega^{m-k}_\bbZ(\Sigma)$. 
Summing up, we obtain
\begin{cor}\label{corCauchyPblDyn}
The embedding $\iota_\Sigma: \Sigma \to M$ of  $\Sigma$ into 
$M$ induces an isomorphism of Abelian groups
\begin{flalign}
\xymatrix{
\Omega^k_\bbZ \cap \ast \Omega^{m-k}_\bbZ(M)
	\ar@<3pt>[rr]^-{(\iota_\Sigma^\ast,\iota_\Sigma^\ast\, \ast^{-1})}
				&&	\Omega^{k,m-k}_\bbZ(\Sigma) \ar@<1pt>[ll]^-{\solve_\Sigma}~~,
}
\end{flalign}
whose inverse $\solve_\Sigma$ is the map assigning to initial data
$(B,\tilde B) \in \Omega^{k,m-k}_\bbZ(\Sigma)$ the corresponding unique solution 
$F \in \Omega^k_\bbZ \cap \ast \Omega^{m-k}_\bbZ(M)$ 
of the Cauchy problem \eqref{eqCauchyPblF}.
\end{cor}

Let us consider the central row of the diagram \eqref{eqConfDiagram}. 
Taking into account also naturality of $\kappa$ and $\cu$, one finds that the diagram of Abelian groups
\begin{flalign}\label{eqCauchyPblDiagram}
\xymatrix{
0 \ar[r]	&	\H^{k-1,m-k-1}(M;\bbT)
				\ar[r]^-{\kappa \times \kappa} \ar[d]_-{\iota_\Sigma^\ast \times \iota_\Sigma^\ast}
				&	\Conf^k(M;\bbZ) \ar[r]^-{\cu_1} \ar[d]_-{\iota_\Sigma^\ast \times \iota_\Sigma^\ast}
						& \Omega^k_\bbZ \cap \ast \Omega^{m-k}_\bbZ(M) 
								\ar[r] \ar[d]^-{(\iota_\Sigma^\ast,\iota_\Sigma^\ast\, \ast^{-1})}	&	0	\\
0 \ar[r]	&	\H^{k-1,m-k-1}(\Sigma;\bbT) \ar[r]_-{\kappa \times \kappa}
				&	\dH^{k,m-k}(\Sigma;\bbZ) \ar[r]_-{\cu \times \cu}
						& \Omega^{k,m-k}_\bbZ(\Sigma) \ar[r]									&	0
}
\end{flalign}
commutes and its rows are short exact sequences.
Using also Lemma \ref{lemChainHomotopy} (ii), Corollary~\ref{corCauchyPblDyn} 
and the five lemma, we obtain
\begin{theo}\label{thmCauchyPbl}
The embedding $\iota_\Sigma: \Sigma \to M$ induces an isomorphism of Abelian groups
\begin{flalign}
\xymatrix{\Conf^k(M;\bbZ) \ar[rr]^-{\iota_\Sigma^\ast \times \iota_\Sigma^\ast}
				&&	\dH^{k,m-k}(\Sigma;\bbZ)~.
}
\end{flalign}
\end{theo}

We can interpret the result of Theorem~\ref{thmCauchyPbl}
as establishing the well-posedness of the initial value problem for $(h,\tilde h) \in \dH^{k,m-k}(M;\bbZ)$ given by
\begin{align}\label{eqCauchyPbl}
\cu\, h = \ast \, \cu\, \tilde h~,	\qquad \iota_\Sigma^\ast h = h_\Sigma~,	\quad	\iota_\Sigma^\ast \tilde h = \tilde h_\Sigma~,
\end{align}
for initial data $(h_\Sigma,\tilde h_\Sigma) \in \dH^{k,m-k}(\Sigma;\bbZ)$.
It follows that the semi-classical configuration space $\Conf^k(M;\bbZ)$ arises as 
the space of solutions of this Cauchy problem.

\begin{rem}\label{remPresymplStructureCompact}
If $M$ has compact Cauchy surfaces $\Sigma$, we can easily endow $\Conf^k(M;\bbZ)$ 
with the structure of a presymplectic Abelian group induced by the ring structure $\cdot$ on differential characters, 
see \cite{CS, SS, BB}. For this, we define the circle-valued presymplectic structure 
\begin{flalign}\label{eqsymptmp}
\sigma :	 \Conf^k(M;\bbZ) \times \Conf^k(M;\bbZ)
\longrightarrow			\bbT	~, \qquad
		\big((h,\tilde h)\,,\, (h', \tilde h{}' \, )\big)		\longmapsto		\big(\iota_\Sigma^\ast(\tilde h \cdot h'
																	- \tilde h{}' \cdot h)\big)[\Sigma]~,
\end{flalign}
where $[\Sigma] \in \H_{m-1}(\Sigma)$ denotes the fundamental class of $\Sigma$. 
Using compatibility between the ring structure on differential characters 
and the natural transformations $\iota$, $\kappa$, $\cu$ and $\ch$, 
one can show that $\sigma$ is in fact independent of the choice of $\Sigma$. 
Fixing any Cauchy surface $\Sigma$ and using the isomorphism given in Theorem \ref{thmCauchyPbl},
the presymplectic structure (\ref{eqsymptmp}) can be induced to initial data and thereby agrees with the
one constructed by \cite{Freed:2006yc,Freed:2006ya} from a Hamiltonian perspective.
However, in contrast to \cite{Freed:2006yc,Freed:2006ya} 
our construction does {\em not} depend on the choice of a Cauchy surface, i.e.\ it is generally covariant.
As we show in Section \ref{secObs}, the assumption of compactness of the Cauchy surfaces
can be dropped, provided that one introduces a suitable support restriction
on the semi-classical gauge fields. 
\end{rem}

%%%%%%%%%%%%%%%%%%%%%%%%%%%%%%%%%%%%%%%%%%%%%%%%
%%%%%%%%%%%%%%%%%%%%%%%%%%%%%%%%%%%%%%%%%%%%%%%%

\section{\label{secSCConf}Dual gauge fields with spacelike compact support}
In this section we introduce and analyze a suitable Abelian group
$\Conf^k_{\rm sc}(M;\bbZ)$ 
of semi-classical gauge fields of spacelike compact support. 
Similarly to the case of the usual quantum field theories on curved spacetimes, 
such as Klein-Gordon theory, the role played by $\Conf^k_{\rm sc}(M;\bbZ)$ will be dual
to that of the semi-classical configuration space $\Conf^k(M;\bbZ)$; in fact,
we shall show in Section~\ref{secObs} that elements in $\Conf^k_{\rm sc}(M;\bbZ)$ 
define functionals (i.e.\ classical observables) on $\Conf^k(M;\bbZ)$ which are 
group characters $\Conf^k(M;\bbZ) \to \bbT$. This dual role of the semi-classical 
gauge fields of spacelike compact support will be reflected mathematically
in the fact that $\Conf^k_{\rm sc}(-;\bbZ): \Loc_m \to \Ab$ is a covariant functor, while
$\Conf^k(-;\bbZ): \Loc_m^\op \to \Ab$ is contravariant.
The correct definition of $\Conf^k_{\rm sc}(M;\bbZ)$ is a 
very subtle point because, in contrast to the standard examples 
like Klein-Gordon theory, the Abelian group $\Conf^k_{\rm sc}(M;\bbZ)$
{\em cannot} be presented as a subgroup of $\Conf^k(M;\bbZ)$, see
Remark \ref{remRelConfNonSub} below.
We give a definition of $\Conf^k_{\rm sc}(M;\bbZ)$ in terms of relative differential cohomology 
and frequently refer to the companion paper \cite{BBSS} 
for further technical details.

\subsection{Semi-classical configuration space}
Let $K\subseteq M$ be a compact subset.
In analogy to \eqref{eqConf}, we define the Abelian group $\Conf^k(M,M \setminus J(K);\bbZ)$ 
of semi-classical gauge fields on $M$ relative to $M \setminus J(K)$ as the pullback
\begin{flalign}\label{eqRelConf}
\xymatrix@C=.5cm{
\Conf^k(M,M \setminus J(K);\bbZ) \ar@{-->}[d] \ar@{-->}[rr]	
		& &	\dH^{m-k}(M,M \setminus J(K);\bbZ) \ar[d]^-\dcu\\
		\dH^k(M,M \setminus J(K);\bbZ) \ar[rr]_-\cu
		& &	\Omega^k(M,M \setminus J(K))
}
\end{flalign}
where $\dH^{p}(M,M \setminus J(K);\bbZ)$ denote the relative differential cohomology groups
and $\Omega^k(M,M \setminus J(K))$ denotes the group of relative differential forms,
see \cite{BB,BBSS} for the definitions and our conventions.
We shall make frequent use of the short exact sequence
\begin{flalign}\label{eqRelConfSES}
{\small{\xymatrix@C=.65cm{
0 \ar[r]	&	\H^{k-1,m-k-1}(M,M \setminus J(K);\bbT)
				\ar[r]^-{\kappa \times \kappa} 
				&	\Conf^k(M,M \setminus J(K);\bbZ)
						\ar[r]^-{\cu_1} 
						& \Omega^k_\bbZ \cap \ast \Omega^{m-k}_\bbZ(M,M \setminus J(K)) 
								\ar[r] 	&	0	\\
} } }
\end{flalign}
for relative semi-classical gauge fields, which immediately follows 
from \cite[Part~II, Section~3.3]{BB} and \cite[Theorem~3.2]{BBSS} 
by imitating the proof of Theorem~\ref{thmConfDiagram}.
\begin{rem}\label{remRelConfNonSub}
One may heuristically think of semi-classical gauge fields on $M$
relative to $M\setminus J(K)$ as fields
on $M$ which ``vanish'' outside of the closed light-cone $J(K)$ of $K$.
However, strictly speaking this interpretation is not correct: There is
a group homomorphism $I : \Conf^k(M,M \setminus J(K);\bbZ) \to \Conf^k(M;\bbZ)$ 
which is induced by the group homomorphisms (denoted with abuse of notation by the same symbols) 
$I : \dH^p(M,M\setminus J(K);\bbZ) \to \dH^p(M,\bbZ)$ that restrict
relative differential characters from relative cycles to cycles by precomposing 
them with the homomorphism $Z_{p-1}(M) \to Z_{p-1}(M,M\setminus J(K))$, cf.\ \cite[Section 3.1]{BBSS}.
By \cite[Remark 3.3]{BBSS} and Theorem \ref{thmRelCauchyPbl} below, 
we observe that the homomorphism $I : \Conf^k(M,M \setminus J(K);\bbZ) \to \Conf^k(M;\bbZ)$ 
is not necessarily injective, which implies that $\Conf^k(M,M \setminus J(K);\bbZ)$ is in general not a subgroup of 
$\Conf^k(M;\bbZ)$. 
\end{rem}

We define the Abelian group $\Conf^k_{\rm sc}(M;\bbZ)$ 
of semi-classical gauge fields of spacelike compact support
by formalizing the intuition that for any element $(h,\tilde h)\in \Conf^k_{\rm sc}(M;\bbZ)$ 
there should exist a sufficiently large compact subset $K \subseteq M$
such that $(h,\tilde h)$ can be represented as an element in $\Conf^k(M,M \setminus J(K);\bbZ)$.
Let us denote by $\mathcal{K}_M$ the directed set of compact subsets of $M$
and notice that the assignment $\Conf^k(M,M \setminus J(-);\bbZ): \mathcal{K}_M  \to \Ab$
is a diagram of shape $\mathcal{K}_M$.\footnote{For this, we use the group homomorphisms
$Z_{p-1}(M,M\setminus J(K^\prime\, )) \to Z_{p-1}(M,M\setminus J(K))$ of relative cycles 
which exist for any $K\subseteq K^\prime$.}
Then the intuition is formalized by taking the colimit of this diagram, i.e.\ we define
the semi-classical gauge fields of spacelike compact support by
\begin{flalign}	\label{eqSCConf}
\Conf^k_\sc(M;\bbZ) := \colim\big(\Conf^k(M,M \setminus J(-);\bbZ): \mathcal{K}_M  \to \Ab\big)~.
\end{flalign}

\begin{rem}\label{remCofinal}
The colimit in \eqref{eqSCConf} can be equally well computed by
restricting to the directed set $\mathcal{K}_\Sigma$ of compact subsets 
of any smooth spacelike Cauchy surface $\Sigma$ of $M$. 
In fact, denoting by $\mathcal{C}_M$ the directed set of closed subsets of $M$, 
one notices that the map $\mathcal{K}_M \to \mathcal{C}_M$, $K \mapsto J(K)$, 
preserves the preorder relation induced by inclusion. 
In particular, we may interpret the functor $\Conf^k(M,M \setminus J(-);\bbZ): \mathcal{K}_M  \to \Ab$ 
as the composition of the functors $\Conf^k(M,M \setminus -;\bbZ): \mathcal{C}_M  \to \Ab$ 
and $J: \mathcal{K}_M \to \mathcal{C}_M$; then $\mathcal{K}_\Sigma \subseteq \mathcal{K}_M$ is cofinal 
with respect to $J: \mathcal{K}_M \to \mathcal{C}_M$.
In fact, for each $K \subseteq M$, we have $J(K) \subseteq J(K_\Sigma)$ for $K_\Sigma = J(K) \cap \Sigma$, 
which is by construction a compact subset of $\Sigma$. This observation provides the isomorphism 
\begin{equation}
\Conf^k_\sc(M;\bbZ) \simeq \colim \big(\Conf^k(M,M \setminus J(-);\bbZ): \mathcal{K}_\Sigma  \to \Ab\big)~.
\end{equation}
\end{rem}

Similarly to Remark \ref{remRelConfNonSub}, there is a group homomorphism (denoted with abuse of notation
by the same symbol)
\begin{flalign}\label{eqn:Imapping}
I : \Conf^k_{\rm sc}(M;\bbZ)  \longrightarrow \Conf^k(M;\bbZ) ~,
\end{flalign}
which is however in general not injective, see \cite[Remark~4.4]{BBSS} 
and Corollary \ref{corSCCauchyPbl} below. Hence semi-classical gauge fields
of spacelike compact support cannot in general be faithfully represented as elements in 
the semi-classical configuration space $\Conf^k(M;\bbZ)$. 

%%%%%%%%%%%%%%%%%%%%%%%%%%%%%%%%%%%%%%%%%%%%%%%%
%%%%%%%%%%%%%%%%%%%%%%%%%%%%%%%%%%%%%%%%%%%%%%%%

\subsection{Cauchy problem}\label{secCauchyPblRestr} 
Consider any compact subset $K \subseteq \Sigma$. 
Taking into account the support property of the Cauchy problem considered in Theorem \ref{thmCauchyPblF} 
and applying arguments similar to those in Section~\ref{secCauchyPbl} to the relative case, 
in particular \eqref{eqZCauchyPblF} and \eqref{eqZCauchyPblDualF} (see also Lemma \ref{lemRelChainHomotopy} (i)), 
one concludes that, given initial data
$(B,\tilde B) \in \OmegaZ^{k,m-k}(\Sigma,\Sigma \setminus K)$, the Cauchy problem \eqref{eqCauchyPblF} 
has a unique solution $F \in \OmegaZ^k \cap \ast\OmegaZ^{m-k}(M,M \setminus J(K))$. 
This observation leads us to the relative version of Corollary~\ref{corCauchyPblDyn}.
\begin{cor}\label{corRelCauchyPblDyn}
The embedding $\iota_\Sigma: \Sigma \to M$ induces an isomorphism of Abelian groups
\begin{flalign}
\xymatrix{
\Omega^k_\bbZ \cap \ast \Omega^{m-k}_\bbZ(M,M \setminus J(K))
	\ar@<3pt>[rr]^-{(\iota_\Sigma^\ast,\iota_\Sigma^\ast \, \ast^{-1})}
				&&	\Omega^{k,m-k}_\bbZ(\Sigma,\Sigma \setminus K)
						\ar@<1pt>[ll]^-{\solve_\Sigma}~,
}
\end{flalign}
whose inverse $\solve_\Sigma$ is the map assigning 
to initial data $(B,\tilde B) \in \Omega^{k,m-k}_\bbZ(\Sigma,\Sigma \setminus K)$ the corresponding 
unique solution $F  \in \Omega^k_\bbZ \cap \ast \Omega^{m-k}_\bbZ(M,M \setminus J(K))$ 
of the Cauchy problem \eqref{eqCauchyPblF}. 
\end{cor}
Using \eqref{eqRelConfSES} and \cite[Theorem 3.2]{BBSS},
and the fact that relative differential cohomology is a functor (in a suitable sense, see \cite[Section 3.1]{BBSS}),
we conclude that the diagram of Abelian groups
\begin{flalign}\label{eqRelCauchyPblDiagram}
{\small{\xymatrix@C=.65cm{
0 \ar[r]	&	\H^{k-1,m-k-1}(M,M \setminus J(K);\bbT)
				\ar[r]^-{\kappa \times \kappa} \ar[d]_-{\iota_\Sigma^\ast \times \iota_\Sigma^\ast}
				&	\Conf^k(M,M \setminus J(K);\bbZ)
						\ar[r]^-{\cu_1} \ar[d]_-{\iota_\Sigma^\ast \times \iota_\Sigma^\ast}
						& \Omega^k_\bbZ \cap \ast \Omega^{m-k}_\bbZ(M,M \setminus J(K)) 
								\ar[r] \ar[d]^-{(\iota_\Sigma^\ast,\iota_\Sigma^\ast\, \ast^{-1})}	&	0	\\
0 \ar[r]	&	\H^{k-1,m-k-1}(\Sigma,\Sigma \setminus K;\bbT) \ar[r]_-{\kappa \times \kappa}
				&	\dH^{k,m-k}(\Sigma,\Sigma \setminus K;\bbZ) \ar[r]_-{\cu \times \cu}
						& \Omega^{k,m-k}_\bbZ(\Sigma,\Sigma \setminus K) \ar[r]				&	0
} } }
\end{flalign}
commutes and its rows are short exact sequences. 
Using also Lemma \ref{lemRelChainHomotopy} (ii), Corollary~\ref{corRelCauchyPblDyn}
and the five lemma, we obtain the relative version of Theorem \ref{thmCauchyPbl}.
\begin{theo}\label{thmRelCauchyPbl}
The embedding $\iota_\Sigma: \Sigma \to M$ induces an isomorphism of Abelian groups
\begin{flalign}
\xymatrix{\Conf^k(M,M \setminus J(K);\bbZ) \ar[rr]^-{\iota_\Sigma^\ast \times \iota_\Sigma^\ast}
				&&	\dH^{k,m-k}(\Sigma,\Sigma \setminus K;\bbZ)~.
}
\end{flalign}
\end{theo}

Taking the colimit of \eqref{eqRelCauchyPblDiagram} over the directed set $\mathcal K_\Sigma$ of compact subsets of $\Sigma$
and recalling Remark \ref{remCofinal} we find that the diagram of Abelian groups
\begin{flalign}\label{eqSCCauchyPblDiagram}
\xymatrix@C=33pt{
0 \ar[r]	&	\H^{k-1,m-k-1}_{\rm sc}(M;\bbT)
				\ar[r]^-{\kappa \times \kappa} \ar[d]_-{\iota_\Sigma^\ast \times \iota_\Sigma^\ast}
				&	\Conf^k_{\rm sc}(M;\bbZ) \ar[r]^-{\cu_1} \ar[d]_-{\iota_\Sigma^\ast \times \iota_\Sigma^\ast}
						& \Omega^k_{{\rm
                                                    sc},\bbZ} \cap
                                                \ast
                                                \Omega^{m-k}_{{\rm sc},\bbZ}(M) 
								\ar[r]
                                                                \ar[d]^-{(\iota_\Sigma^\ast,\iota_\Sigma^\ast
                                                                  \, \ast^{-1})}	&	0	\\
0 \ar[r]	&	\H^{k-1,m-k-1}_{\rm c}(\Sigma;\bbT) \ar[r]_-{\kappa \times \kappa}
				&	\dH^{k,m-k}_{\rm c}(\Sigma;\bbZ) \ar[r]_-{\cu \times \cu}
						& \Omega^{k,m-k}_{{\rm
                                                    c},\bbZ}(\Sigma) \ar[r]								&	0
}
\end{flalign}
commutes, its rows are short exact sequences and its vertical arrows are isomorphisms. 
The subscript $_{\rm c}$ denotes compact support
and the various groups of this diagram are {\em defined} by these colimits.\footnote{For a 
detailed presentation of differential characters with compact support, see \cite[Section 4]{BBSS}.} 
This shows that $\Conf^k_{\rm sc}(M;\bbZ)$ is the space of solutions
of the Cauchy problem \eqref{eqCauchyPbl} for $(h,\tilde h)\in \dH^{k,m-k}_{\rm sc}(M;\bbZ)$ 
with initial data in $\dH^{k,m-k}_{\rm c}(\Sigma;\bbZ)$.
\begin{cor}\label{corSCCauchyPbl}
The embedding $\iota_\Sigma: \Sigma \to M$ induces an isomorphism of Abelian groups
\begin{flalign}
\xymatrix{\Conf^k_{\rm sc}(M;\bbZ) \ar[rr]^-{\iota_\Sigma^\ast \times \iota_\Sigma^\ast}
				&&	\dH^{k,m-k}_{\rm c}(\Sigma;\bbZ)~.
}
\end{flalign}
\end{cor}

The assignment of the 
Abelian groups $\Conf_{\rm sc}^k(M;\bbZ)$ to objects $M$ in $\Loc_m$ 
is a covariant functor
\begin{flalign}
\Conf^k_{\rm sc}(-;\bbZ): \Loc_m \longrightarrow	 \Ab~.
\end{flalign}
The group homomorphism $f_\ast := \Conf^k_{\rm sc}(f;\bbZ) : \Conf^k_{\rm sc}(M;\bbZ) 
\to \Conf^k_{\rm sc}(M^\prime;\bbZ)$ associated with a morphism 
$f : M\to M^\prime$ in $\Loc_m$ is constructed in Lemma \ref{lemSCPushforward}.

\begin{rem}\label{remSequencesSCnatural}
With a similar construction as in Lemma \ref{lemSCPushforward}, we obtain two more functors
\begin{flalign}
\H^{k-1,m-k-1}_{\rm sc}(-;\bbT) : \Loc_m \longrightarrow \Ab~~,
\qquad \Omega^k_{{\rm sc},\bbZ} \cap \ast \Omega^{m-k}_{{\rm
    sc},\bbZ}(-) :   \Loc_m \longrightarrow \Ab ~.
\end{flalign}
Using these constructions one can further show that the 
short exact sequence in the first row of the diagram \eqref{eqSCCauchyPblDiagram}
is natural, i.e.\ for any morphism $f : M\to M^\prime$  in $\Loc_m$, the diagram
of Abelian groups
\begin{flalign}
\xymatrix@C=33pt{
0 \ar[r]	&	\H^{k-1,m-k-1}_{\rm sc}(M;\bbT)
				\ar[r]^-{\kappa \times \kappa} \ar[d]_-{f_{\ast}}
				&	\Conf^k_{\rm sc}(M;\bbZ) \ar[r]^-{\cu_1} \ar[d]_-{f_{\ast}}
						& \Omega^k_{{\rm
                                                    sc},\bbZ} \cap
                                                \ast
                                                \Omega^{m-k}_{{\rm sc},\bbZ}(M) 
								\ar[r] \ar[d]^-{f_{\ast}}	&	0	\\
0 \ar[r]	&	\H^{k-1,m-k-1}_{\rm sc}(M^\prime;\bbT)
				\ar[r]_-{\kappa \times \kappa} 
				&	\Conf^k_{\rm sc}(M^\prime;\bbZ) \ar[r]_-{\cu_1} 
						& \Omega^k_{{\rm
                                                    sc},\bbZ} \cap
                                                \ast
                                                \Omega^{m-k}_{{\rm sc},\bbZ}(M^\prime\, ) 
								\ar[r] 	&	0	
}
\end{flalign}
commutes. Here we also use the notation $f_\ast$ for the group homomorphisms 
$\H^{k-1,m-k-1}_{\rm sc}(f;\bbT) $ and $\Omega^k_{{\rm sc},\bbZ} \cap
\ast \Omega^{m-k}_{{\rm sc},\bbZ}(f) $.
\end{rem}

%%%%%%%%%%%%%%%%%%%%%%%%%%%%%%%%%%%%%%%%%%%%%%%%
%%%%%%%%%%%%%%%%%%%%%%%%%%%%%%%%%%%%%%%%%%%%%%%%

\section{\label{secObs}Observables for dual gauge fields}
In this section we introduce and analyze a suitable Abelian group
$\Obs^k(M;\bbZ)$ of basic semi-classical observables.
In general, observables are given by functionals on the configuration
space of the field theory.
Recalling that the semi-classical configuration space $\Conf^k(M;\bbZ)$ is an Abelian group,
there is a distinguished class of observables given by the
group characters $\Conf^k(M;\bbZ)^\star := \Hom_{\Ab}(\Conf^k(M;\bbZ),\bbT)$.
However, generic group characters define observables that are too singular
for quantization, hence it is reasonable to impose a suitable regularity condition
in the spirit of smooth Pontryagin duality \cite{HLZ, Becker:2014tla}. 
After defining and analyzing the smooth Pontryagin dual $\Obs^k(M;\bbZ)$ of $\Conf^k(M;\bbZ)$,
we shall show that it is isomorphic to the Abelian group $\Conf^k_{\rm sc}(M;\bbZ)$
of semi-classical gauge fields of spacelike compact support.
Generalizing the constructions of Remark \ref{remPresymplStructureCompact}
to the case of not necessarily compact Cauchy surfaces, we obtain a natural presymplectic structure
on the Abelian group of semi-classical observables $\Obs^k(M;\bbZ)$.
We shall analyze properties of these presymplectic Abelian groups
in view of the axioms of locally covariant quantum field theory \cite{Brunetti:2001dx}.

\subsection{Semi-classical observables}
We shall begin by imposing a suitable regularity condition on the 
Abelian group of group characters $\Conf^k(M;\bbZ)^\star $ 
of the semi-classical configuration space.

\begin{defi}\label{defSmoothDual}
The Abelian group $\Obs^k(M;\bbZ)$ of semi-classical observables
is the following subgroup of $\Conf^k(M;\bbZ)^\star$:
A group character $\varphi \in \Conf^k(M;\bbZ)^\star$ is an element in 
$\Obs^k(M;\bbZ)$ if and only if there exists 
$\omega = \ast \tilde \omega \in \Omega^k_{{\rm sc},\bbZ} \cap \ast\Omega^{m-k}_{{\rm sc},\bbZ}(M)$ 
and a smooth spacelike Cauchy surface $\Sigma$ of $M$ such that
\begin{flalign}\label{eqSmoothDual}
\varphi\big((\iota \times \iota)([A],[\tilde A])\big) 
= \int_\Sigma \, \big( \tilde A \wedge \omega - (-1)^{k\, (m-k)} \, A \wedge \tilde \omega\big) \mod \bbZ~, 
\end{flalign}
for all  semi-classical topologically trivial fields 
$([A],[\tilde A]) \in \TT^k(M;\bbZ)$.
\end{defi}

We now prove that Definition~\ref{defSmoothDual} does not depend on the choice
of Cauchy surface $\Sigma$ used to 
evaluate the integral (\ref{eqSmoothDual}).
For this, notice that
$\omega = \ast \tilde \omega \in \Omega^k_{{\rm sc},\bbZ} \cap \ast\Omega^{m-k}_{{\rm sc},\bbZ}(M)$ 
implies $\Box \, \omega =0$ and $\Box \, \tilde\omega =0$, where $\Box
= \de \, \dd + \dd\, \de$
is the d'Alembert operator. By \cite{Bar:2007zz, Bar} there exists
$\tilde\beta \in\Omega^{m-k}_{\rm c}(M)$ such 
that $\tilde\omega = G\tilde\beta$, where
$G = G^+ - G^-$ is the causal propagator and $G^\pm$ are the retarded/advanced Green's operators
of $\Box$.
We further have $\omega = \ast\tilde \omega = G{\ast}\tilde\beta$.
Because of $\dd \omega =0$ and $\dd\tilde\omega=0$,
there exist $\alpha\in\Omega^{k+1}_{\rm c}(M)$ and $\tilde\alpha\in\Omega^{m-k+1}_{\rm c}(M)$ 
such that $\dd {\ast}\tilde \beta = \Box \alpha$ and $\dd\tilde\beta = \Box\tilde\alpha$.
Using these observations, and realizing $\Sigma$ as the boundary of $J^-(\Sigma)\subseteq M$
and also as the boundary of $J^+(\Sigma)\subseteq M$ (with opposite orientation),
we can rewrite \eqref{eqSmoothDual} as
\begin{flalign}
\nn \varphi\big((\iota \times \iota)([A],[\tilde A])\big)  &=
\int_{\Sigma}\, \big(
 \tilde A \wedge G{\ast}\tilde\beta - (-1)^{k\, (m-k)}\, A \wedge
 G\tilde\beta \, \big) \mod \bbZ\\[4pt]
\nn &=\int_{J^-(\Sigma)} \, \dd\big( \tilde A \wedge
G^+{\ast}\tilde\beta - (-1)^{k\, (m-k)} \, A \wedge  G^+\tilde\beta \,
\big)\mod\bbZ\\
\nn &~~~~~~~~+ \int_{J^+(\Sigma)} \, \dd\big( \tilde A \wedge
G^-{\ast}\tilde\beta - (-1)^{k\, (m-k)} \, A \wedge  G^-\tilde\beta \,
\big)\mod\bbZ\\[4pt]
&= \int_M\, \big( (-1)^{m-k}\,\tilde A \wedge \alpha 
- (-1)^{k\, (m-k)}\,(-1)^k \, A\wedge \tilde\alpha\big)\mod\bbZ~,\label{eqSmoothDualspacetime}
\end{flalign}
where we have also used $\dd A = \ast \dd \tilde A$. It then follows that \eqref{eqSmoothDual} is independent
of the choice of Cauchy surface because \eqref{eqSmoothDualspacetime} shows that it can be written
as an integral over spacetime $M$.

\begin{rem}\label{remSmoothDual}
There is an alternative but equivalent definition of the
Abelian group $\Obs^k(M;\bbZ)$ of semi-classical observables 
which employs the notion of smooth Pontryagin duality developed in \cite{HLZ, Becker:2014tla}. 
Taking the smooth Pontryagin dual of the pullback diagram \eqref{eqConf} which defines
the semi-classical configuration space $\Conf^k(M;\bbZ)$, we define
an Abelian group $\Conf^k(M;\bbZ)^\star_\infty$ (called the smooth Pontryagin dual of $\Conf^k(M;\bbZ)$)
via the pushout
\begin{flalign}
\xymatrix{
\Omega_{\rm c}^k(M) \ar[d]_-{\cu^\star} \ar[r]^-{(\dcu)^\star}	&	\dH^{m-k}(M;\bbZ)^\star_{\infty} \ar@{-->}[d]			\\
\dH^k(M;\bbZ)^{\star}_{\infty} \ar@{-->}[r]					&	\Conf^k(M;\bbZ)^\star_\infty
}
\end{flalign}
where $\dH^p(M;\bbZ)^{\star}_{\infty}$ denotes the smooth Pontryagin dual of $\dH^p(M;\bbZ)$.
This pushout may be realized explicitly as the quotient
\begin{flalign}\label{eqSmoothDualIso}
\Conf^k(M;\bbZ)^\star_\infty 
=  \frac{\dH^{k}(M;\bbZ)^\star_\infty \oplus \dH^{m-k}(M;\bbZ)^\star_\infty }{\big\{ \cu^\star\omega \oplus - (\dcu)^\star\omega \, :\,  \omega \in \Omega_{\rm c}^k(M)\big\}}~.
\end{flalign}
One can show that the Abelian group $\Conf^k(M;\bbZ)^\star_\infty $ is isomorphic
to the Abelian group $\Obs^k(M;\bbZ)$ of semi-classical observables given in 
Definition \ref{defSmoothDual}. As we do not need this isomorphism in
this paper, we refrain from writing it out explicitly. Let us just point out that 
the elements of the smooth Pontryagin dual are in particular continuous group characters. 
In fact, on account of \cite[Appendix A]{Becker:2014tla}, 
all differential cohomology groups on a manifold of finite type are Fr\'echet-Lie groups 
that are (non-canonically) isomorphic to the Cartesian product of a torus, a torsion group, 
a discrete lattice in a Euclidean space (all finite dimensional) and a Fr\'echet vector space of differential forms. 
This observation allows one to conclude that the elements of the smooth Pontryagin dual 
are continuous group characters with respect to the Fr\'echet topology mentioned above. 
\end{rem}

We now show that the assignment of the Abelian groups $\Obs^k(M;\bbZ)$ of semi-classical observables
to objects $M$ in $\Loc_m$ is a covariant functor
\begin{flalign}\label{eqObsFunctor}
\Obs^k(-;\bbZ): \Loc_m \longrightarrow \Ab~.
\end{flalign}
For this, note that
the assignment of character groups $\Conf^k(M;\bbZ)^\star = \Hom_{\Ab}(\Conf^k(M;\bbZ),\bbT)$
(without the regularity condition of Definition \ref{defSmoothDual})
to objects $M$ in $\Loc_m$ is a covariant functor 
$\Conf^k(-;\bbZ)^\star : \Loc_m\to\Ab$: Given any morphism
$f : M\to M^\prime$ in $\Loc_m$, functoriality of the semi-classical configuration spaces
provides us with a group homomorphism $f^\ast = \Conf^k(f;\bbZ): \Conf^k(M^\prime;\bbZ)\to \Conf^k(M;\bbZ)$,
which we can dualize to a group homomorphism (called pushforward)
$f_\ast :=\Conf^k(f;\bbZ)^\star  = (f^\ast)^\star: \Conf^k(M;\bbZ)^\star\to \Conf^k(M^\prime;\bbZ)^\star$
between the character groups. It remains to show that these group homomorphisms
induce group homomorphisms $f_\ast : \Obs^k(M;\bbZ)\to \Obs^k(M^\prime;\bbZ)$, i.e.\ 
that pushforwards preserve the regularity condition of Definition  \ref{defSmoothDual}.
Let $\varphi \in \Obs^k(M;\bbZ)$ and 
$\omega = \ast \tilde \omega \in \Omega^k_{{\rm sc},\bbZ} \cap \ast \Omega^{m-k}_{{\rm sc},\bbZ}(M)$ 
be as in Definition \ref{defSmoothDual}. 
Exploiting the Cauchy problem described by Corollary \ref{corRelCauchyPblDyn},
we can easily push forward $\omega$ and $\tilde\omega$ to
$f_\ast \omega = \ast f_\ast \tilde\omega \in \Omega^k_{{\rm sc},\bbZ} \cap \ast \Omega^{m-k}_{{\rm sc},\bbZ}(M^\prime\, )$
by pushing forward the initial data from a Cauchy surface $\Sigma\subseteq M$
to a suitable Cauchy surface $\Sigma^\prime\subseteq M^\prime$.\footnote{
A suitable Cauchy surface $\Sigma^\prime\subseteq M^\prime$ can be constructed follows: 
Let $U\subseteq \Sigma$ be an open neighborhood
of $\supp\,\omega \cap \Sigma$ with compact closure. 
Then the image $f(\, \overline{U}\, )\subseteq M^\prime$ of the closure $\overline{U}$ of $U$ 
is a spacelike and acausal compact submanifold with boundary, and we can take
$\Sigma^\prime$ to be any Cauchy surface extending $f(\,
\overline{U}\, )$; see \cite[Theorem 1.1]{Bernal:2005qf}
for the existence of such a Cauchy surface.
}
 By construction, we have
\begin{flalign}
\nn f_\ast \varphi\big((\iota \times \iota)([A],[\tilde A])\big) 
&	= \varphi\big((\iota \times \iota)([f^\ast A],[f^\ast \tilde A])\big) \\[4pt]
\nn &	= \int_\Sigma\, \big( f^\ast \tilde A \wedge \omega -
(-1)^{k\, (m-k)} \, f^\ast A \wedge \tilde \omega\big) \mod \bbZ \\[4pt]
&	= \int_{\Sigma^\prime}\, \big(\tilde A \wedge f_\ast \omega 
	- (-1)^{k\, (m-k)}\, A \wedge f_\ast \tilde \omega\big) \mod \bbZ~,
\end{flalign}
which shows that $f_\ast \varphi\in \Obs^k(M^\prime;\bbZ)$ as required.

%%%%%%%%%%%%%%%%%%%%%%%%%%%%%%%%%%%%%%%%%%%%%%%%
%%%%%%%%%%%%%%%%%%%%%%%%%%%%%%%%%%%%%%%%%%%%%%%%

\subsection{\label{subObsSCConf}Observables from spacelike compact gauge fields}
We shall now show that the Abelian group $\Conf^k_{\rm sc}(M;\bbZ)$ of
semi-classical gauge fields of spacelike compact support is isomorphic to the Abelian group
$\Obs^k(M;\bbZ)$ of semi-classical observables introduced in Definition  \ref{defSmoothDual}.
Using the techniques which allow us to establish this isomorphism,
we shall also prove that $\Obs^k(M;\bbZ)$ is large enough to separate points of the semi-classical configuration space
$\Conf^k(M;\bbZ)$, i.e.\ for $(h,\tilde h), (h',\tilde h'\, )\in \Conf^k(M;\bbZ)$ we have
$\varphi\big((h,\tilde h)\big) = \varphi\big((h',\tilde h'\, )\big)$ for all 
$\varphi \in \Obs^k(M;\bbZ)$ if and only if $(h,\tilde h) = (h',\tilde h'\,)$.
\sk 

By~\cite[Section~5.2]{BBSS}, for any smooth
spacelike Cauchy surface $\Sigma$ of $M$ there is a $\bbT$-valued pairing 
$\ips{\cdot}{\cdot}_{\rm c} : \dH^{m-p}(\Sigma;\bbZ) \times \dH^p_{\rm c}(\Sigma;\bbZ)\to\bbT$
between differential cohomology and compactly supported differential cohomology. Using the isomorphisms
given in Theorem~\ref{thmCauchyPbl} and Corollary~\ref{corSCCauchyPbl}, we define
a $\bbT$-valued pairing between $\Conf^k(M;\bbZ)$ and $\Conf^k_{\rm sc}(M;\bbZ)$ by
\begin{flalign}
\nn \ips{\cdot}{\cdot}:	\Conf^k(M;\bbZ) \times \Conf^k_{\rm sc}(M;\bbZ)	&	\longrightarrow \bbT~,						\\
\big((h,\tilde h)\,,\,(h',\tilde h'\,)\big)	&	\longmapsto
\ips{\iota_\Sigma^\ast \tilde h}{\iota_\Sigma^\ast h'\, }_{\rm c}- 
(-1)^{k\,(m-k)}\,\ips{\iota_\Sigma^\ast h}{\iota_\Sigma^\ast\tilde h'\,}_{\rm c}~.		\label{eqPairingConf}
\end{flalign}
In Lemma \ref{lemPairingSigmaindep} we show that this pairing does not depend on
the choice of Cauchy surface $\Sigma$ and we prove its naturality in the sense that
for any morphism $f: M \to M^\prime$ in $\Loc_m$ the diagram of Abelian groups
\begin{flalign}
\xymatrix{
\ar[d]_-{\id\times f_\ast}\Conf^k(M^\prime;\bbZ) \times \Conf^k_{\rm sc}(M;\bbZ) \ar[rr]^-{f^\ast \times \id} && \Conf^k(M;\bbZ) \times \Conf^k_{\rm sc}(M;\bbZ)\ar[d]^-{\ips{\cdot}{\cdot}}\\
\Conf^k(M^\prime;\bbZ) \times \Conf^k_{\rm sc}(M^\prime;\bbZ)\ar[rr]_-{\ips{\cdot}{\cdot}} && \bbT
}
\end{flalign}
commutes. 
\sk

By partial evaluation, the pairing \eqref{eqPairingConf} allows us to define group characters on 
$\Conf^k(M;\bbZ)$: For any $(h',\tilde h'\, )\in \Conf^k_{\rm sc}(M;\bbZ)	$
there is a group character
\begin{flalign}\label{eqpartialeval}
\ips{\cdot\, }{\, (h',\tilde h'\, )} : \Conf^k(M;\bbZ)\longrightarrow
\bbT~,~~(h,\tilde h)\longmapsto \ips{(h,\tilde h)\, }{\, (h' ,\tilde h'\, )}~.
\end{flalign}
The next result in particular allows us to separate points of the semi-classical configuration space
$\Conf^k(M;\bbZ)$ by using only such group characters.
\begin{propo}\label{prpNonDeg}
The pairing $\ips{\cdot}{\cdot}$ introduced in \eqref{eqPairingConf} is weakly non-degenerate. 
\end{propo}
\begin{proof}
Recalling Theorem \ref{thmCauchyPbl} and Corollary \ref{corSCCauchyPbl}, 
the pullback along the embedding $\iota_\Sigma:\Sigma\to M$ provides isomorphisms 
$\Conf^k(M;\bbZ) \simeq \dH^{k,m-k}(\Sigma;\bbZ)$ 
and $\Conf^k_{\rm sc}(M;\bbZ) \simeq \dH^{k,m-k}_{\rm c}(\Sigma;\bbZ)$. 
Using these isomorphisms, the pairing \eqref{eqPairingConf} corresponds precisely
to the pairing
\begin{flalign}
\nn \ips{\cdot}{\cdot}_{\Sigma}^{} : \dH^{k,m-k}(\Sigma;\bbZ) \times \dH^{k,m-k}_{\rm c}(\Sigma;\bbZ)	&	\longrightarrow	\bbT~,	\\
\big((h,\tilde h)\,,\,(h',\tilde h'\, )\big)	&	\longmapsto
\ips{\tilde h}{h'\, }_{\rm c} - (-1)^{k\, (m-k)} \, \ips{h}{\tilde
  h'\, }_{\rm c} ~\label{eqPairingBiDiffChar}
\end{flalign}
 between initial data on $\Sigma$. The proof then follows from weak non-degeneracy
 of the pairing $\ips{\cdot}{\cdot}_{\rm c} : \dH^{m-p}(\Sigma;\bbZ) \times \dH^p_{\rm c}(\Sigma;\bbZ)\to\bbT$, 
 cf.\ \cite[Corollary~5.6]{BBSS}.
\end{proof}

Finally, we show that the partial evaluation \eqref{eqpartialeval} 
establishes an isomorphism between $\Conf^k_{\rm sc}(M;\bbZ)$
and the Abelian group $\Obs^k(M;\bbZ)$ of semi-classical
observables introduced in Definition~\ref{defSmoothDual}.
\begin{propo}\label{prpSCConfObsIso}
The group homomorphism
\begin{flalign}\label{eqtmphom}
\obs: \Conf^k_{\rm sc}(M;\bbZ)		 \longrightarrow
\Obs^k(M;\bbZ)~, \qquad 
(h',\tilde h'\, )					\longmapsto
\ips{\cdot\, }{\, (h',\tilde h'\, )}~,
\end{flalign}
is an isomorphism which provides a natural isomorphism
between the functors $\Conf^k_{\rm sc}(-;\bbZ): \Loc_m \to \Ab$ 
and $\Obs^k(-;\bbZ): \Loc_m \to \Ab$. 
\end{propo}
\begin{proof}
We first have to show that the group character
$\ips{\cdot\, }{\, (h',\tilde h'\, )}$ satisfies the regularity condition of Definition \ref{defSmoothDual},
for any $(h',\tilde h'\, ) \in \Conf^k_{\rm sc}(M;\bbZ)$. This follows
from \cite[eq.\ (4.12)]{BBSS}: For any $([A],[\tilde A]) \in \TT^k(M)$, we find 
\begin{flalign}
\ip{(\iota \times \iota)([A],[\tilde A])\, }{\, (h',\tilde h'\, )} 
= \int_\Sigma\, \big(\tilde A \wedge \cu\, h' - (-1)^{k\, (m-k)}\, A
\wedge \cu\, \tilde h'\, \big) \mod \bbZ~.
\end{flalign}
Moreover, \eqref{eqtmphom} is injective due to Proposition \ref{prpNonDeg}. 
\sk

To show that \eqref{eqtmphom} is surjective, 
let us take any $\varphi \in \Obs^k(M;\bbZ)$ and choose a smooth spacelike Cauchy surface $\Sigma$ of $M$. 
Using the isomorphism established in Theorem \ref{thmCauchyPbl} and recalling Definition \ref{defSmoothDual},
there exists a unique smooth character $\varphi_\Sigma \in \dH^{k,m-k}(\Sigma;\bbZ)^\star_\infty$ 
such that $\varphi_\Sigma \circ (\iota_\Sigma^\ast \times \iota_\Sigma^\ast) = \varphi$.
Using further the character duality proven in \cite[Theorem~5.4]{BBSS},
there exists a unique $(h'_\Sigma,\tilde h_\Sigma') \in \dH^{k,m-k}_{\rm c}(\Sigma;\bbZ)$ 
such that $\ips{\cdot\, }{\, (h'_\Sigma,\tilde h'_\Sigma)}_{\Sigma}^{}
= \varphi_\Sigma$ by \eqref{eqPairingBiDiffChar}. 
Using Corollary~\ref{corSCCauchyPbl}, we introduce $(h',\tilde h'\, ) \in \Conf^k_{\rm sc}(M;\bbZ)$ as the unique 
element with $(\iota_\Sigma^\ast \times \iota_\Sigma^\ast)(h',\tilde
h'\, ) = (h'_\Sigma,\tilde h'_\Sigma)$. 
Then the definition of $\ips{\cdot}{\cdot}$ given in
\eqref{eqPairingConf} implies that $\varphi = \ips{\cdot\, }{\,
  (h',\tilde h'\, )}$. 
\sk

It remains to prove that the established isomorphism is natural.
Let $f: M \to M^\prime$ be a morphism in $\Loc_m$.
Recall that the pushforward $f_\ast: \Obs^k(M;\bbZ) \to \Obs^k(M^\prime;\bbZ)$ is given by the Pontryagin dual 
of the pullback $f^\ast: \Conf^k(M^\prime;\bbZ) \to \Conf^k(M;\bbZ)$
and that Lemma \ref{lemPairingSigmaindep} establishes naturality 
of the pairing $\ips{\cdot}{\cdot}$. Therefore, for all $(h',\tilde
h'\, ) \in \Conf^k_{\rm sc}(M;\bbZ)$, one has 
$f_\ast\ips{\cdot\, }{\, (h',\tilde h'\, )} = \ips{f^\ast \cdot\, }{\,
  (h',\tilde h'\, )} = \ips{\cdot\, }{\, f_\ast(h',\tilde h'\, )}$, 
i.e.\ the diagram
\begin{flalign}
\xymatrix{
\Conf^k_{\rm sc}(M;\bbZ) \ar[r]^-{\obs} \ar[d]_-{f_\ast}	&	\Obs^k(M;\bbZ) \ar[d]^-{f_\ast}		\\
\Conf^k_{\rm sc}(M^\prime;\bbZ) \ar[r]_-{\obs}			&	\Obs^k(M^\prime;\bbZ)
}
\end{flalign}
commutes and hence \eqref{eqtmphom} is a natural isomorphism.
\end{proof}

%%%%%%%%%%%%%%%%%%%%%%%%%%%%%%%%%%%%%%%%%%%%%%%%
%%%%%%%%%%%%%%%%%%%%%%%%%%%%%%%%%%%%%%%%%%%%%%%%

\subsection{\label{subSympl}Presymplectic structure}
We will introduce a natural $\bbT$-valued presymplectic structure $\tau$
on the Abelian group $\Obs^k(M;\bbZ)$ of semi-classical observables.
In this way we obtain a functor $(\Obs^k(-;\bbZ),\tau): \Loc_m \to
\PSAb$ valued in the category $\PSAb$
of presymplectic Abelian groups (with group homomorphisms preserving the presymplectic structures as morphisms). 
This will be the main input for Section \ref{secQuantization}, 
where the quantization of the semi-classical model described by $(\Obs^k(-;\bbZ),\tau)$ will be addressed. 
\begin{propo}\label{prpPreSympl}
Let $\obs: \Conf^k_{\rm sc}(M;\bbZ) \to \Obs^k(M;\bbZ)$ be the isomorphism
introduced in Proposition \ref{prpSCConfObsIso} and $I : \Conf^k_{\rm
  sc}(M;\bbZ) \to \Conf^k(M;\bbZ) $ the group homomorphism
given in \eqref{eqn:Imapping}. Then
\begin{flalign}
\tau: \Obs^k(M;\bbZ) \times \Obs^k(M;\bbZ)
\longrightarrow			\bbT~, \qquad							
(\varphi,\varphi'\, )
\longmapsto		\ips{I(\obs^{-1} \varphi)}{\obs^{-1}
  \varphi'\, }
\end{flalign}
defines a presymplectic structure on $\Obs^k(M;\bbZ)$ whose radical is $\obs(\ker I)$. 
\end{propo}
\begin{proof}
Up to the isomorphism $\obs^{-1}: \Obs^k(M;\bbZ) \to \Conf^k_{\rm sc}(M;\bbZ)$, 
the mapping $\tau$ is given by 
\begin{flalign}\label{eqPreSymplSCConf}
\sigma: \Conf^k_{\rm sc}(M;\bbZ) \times \Conf^k_{\rm sc}(M;\bbZ)
&	\longrightarrow			\bbT~, \qquad
\big((h,\tilde h)\,,\,(h',\tilde h'\, )\big)					\longmapsto		\ips{I(h,\tilde h)\,}{(\,h',\tilde h'\,)}~.
\end{flalign}
Recalling \eqref{eqPairingConf}, we observe that $\sigma$ is
$\bbZ$-bilinear and hence so is $\tau$. 
Antisymmetry follows from \cite[Proposition~5.9]{BBSS}.
Since the pairing $\ips{\cdot}{\cdot}: \Conf^k(M;\bbZ) \times \Conf^k_{\rm sc}(M;\bbZ) \to \bbT$ 
is weakly non-degenerate by Proposition~\ref{prpNonDeg}, the radical of
$\sigma$ coincides with the kernel of $I$ which implies that the radical of $\tau$ is $\ker(I \circ \obs^{-1}) = \obs(\ker I)$. 
\end{proof}
\begin{rem}
If $M$ has compact Cauchy surfaces $\Sigma \subseteq M$, the group homomorphism
$I : \Conf^k_{\rm sc}(M;\bbZ) \to \Conf^k(M;\bbZ) $
is the identity. In fact, with $\Sigma$ compact, $J(\Sigma) = M$ entails 
that the diagram whose colimit defines $\Conf^k_{\rm sc}(M;\bbZ)$ 
has $\Conf^k(M,M \setminus J(\Sigma);\bbZ) = \Conf^k(M;\bbZ)$ as its terminal object. 
In particular, $\tau$ is actually weakly symplectic for globally hyperbolic Lorentzian manifolds with compact Cauchy surfaces. 
In this case \eqref{eqPreSymplSCConf} coincides with the (pre)symplectic structure
described in Remark \ref{remPresymplStructureCompact}.
\end{rem}

Our next task is to prove that the presymplectic structure $\tau$ introduced in Proposition \ref{prpPreSympl} is natural, 
so that we can interpret $(\Obs^k(-;\bbZ),\tau)$ as a functor from $\Loc_m$ to $\PSAb$. 
\begin{propo}\label{prpPreSymplNat}
Let $f: M \to M^\prime$ be a morphism in $\Loc_m$.
Then the diagram of Abelian groups 
\begin{flalign}
\xymatrix{
\Obs^k(M;\bbZ) \times \Obs^k(M;\bbZ) \ar[dr]^-\tau \ar[dd]_-{f_\ast\times f_{\ast}}
																				&	\\ 
				&	\bbT 	\\
\Obs^k(M^\prime;\bbZ) \times \Obs^k(M^\prime;\bbZ) \ar[ur]_-\tau				&	
}
\end{flalign}
commutes.
\end{propo}
\begin{proof}
Recalling from Proposition \ref{prpSCConfObsIso} that $\obs: \Conf^k_{\rm sc}(-;\bbZ) \Rightarrow \Obs^k(-;\bbZ)$ 
is a natural isomorphism, it is enough to prove commutativity of the diagram 
\begin{flalign}
\xymatrix{
\Conf^k_{\rm sc}(M;\bbZ) \times \Conf^k_{\rm sc}(M;\bbZ)
	\ar[dr]^-\sigma  \ar[dd]_-{f_\ast\times f_{\ast}}			&	\\
	&  \bbT 	\\
\Conf^k_{\rm sc}(M^\prime;\bbZ) \times \Conf^k_{\rm sc}(M^\prime;\bbZ) \ar[ur]_-\sigma	&	
}
\end{flalign}
for $\sigma$ given in \eqref{eqPreSymplSCConf}. Making use of naturality of the pairing 
$\ips{\cdot}{\cdot}$ (cf.\ Lemma \ref{lemPairingSigmaindep}), we obtain
\begin{flalign}
\nn \sigma\big(f_\ast(h,\tilde h)\,,\, f_\ast(h',\tilde h'\, )\big) &= \ips{I\, f_\ast(h,\tilde h)\, }{\, f_\ast(h',\tilde h'\,)} \\[4pt]
&= \ips{f^\ast\, I \,f_\ast(h,\tilde h)\, }{\, (h',\tilde h'\, )} \nn \\[4pt]
&= \ips{I (h,\tilde h)\, }{\, (h',\tilde h'\, )} = \sigma\big((h,\tilde h)\,,\, (h',\tilde h'\, )\big)~,
\end{flalign}
for all $(h,\tilde h), (h,\tilde h'\, )\in \Conf^k_{\rm sc}(M;\bbZ)$.
In the third equality we used $f^\ast\circ I\circ f_\ast = I$ which is proven in Lemma~\ref{lemIproperties}.
\end{proof}

\subsection{\label{subPropPSAb}Locally covariant field theory}
We analyze properties of the functor $(\Obs^k(-;\bbZ),\tau): \Loc_m \to \PSAb$ from the point of view 
of the axioms of locally covariant field theory \cite{Brunetti:2001dx}. 
\begin{propo}[Causality axiom]\label{prpCausality}
Let $M_1 \stackrel{f_1}{\longrightarrow} M \stackrel{f_2}{\longleftarrow} M_2$ be
a diagram in $\Loc_m$ such that the images of $f_1$ and $f_2$ are causally disjoint, i.e.\
$J(f_1(M_1)) \cap f_2(M_2) = \emptyset$. 
Then the presymplectic structure $\tau: \Obs^k(M;\bbZ) \times \Obs^k(M;\bbZ) \to \bbT$ 
vanishes on ${f_1}_\ast(\Obs^k(M_1;\bbZ)) \times {f_2}_\ast(\Obs^k(M_2;\bbZ))$. 
\end{propo}
\begin{proof}
Using again the natural isomorphism
$\obs: \Conf^k_{\rm sc}(-;\bbZ) \Rightarrow \Obs^k(-;\bbZ)$,
it is equivalent to prove the analogous statement for $\sigma$ given in \eqref{eqPreSymplSCConf}.
Given $(h,\tilde h)\in \Conf^k_{\rm sc}(M_1;\bbZ)$ and $(h',\tilde h'\, )\in \Conf^k_{\rm sc}(M_2;\bbZ)$,
naturality of the pairing $\ips{\cdot}{\cdot}$ implies
\begin{flalign}
\sigma\big({f_1}_\ast (h,\tilde h) \,,\, {f_{2}}_\ast(h',\tilde h'\, )\big)= 
\ips{f_{2}^\ast \, I \, {f_{1}}_\ast (h,\tilde h)\, }{\, (h',\tilde h'\, )}
\end{flalign}
and the proof follows from $f_{2}^\ast\circ I \circ {f_{1}}_\ast =0$, see Lemma \ref{lemIproperties}.
\end{proof}

\begin{propo}[Time-slice axiom]\label{prpTimeSlice}
Let $f: M \to M^\prime$ be a Cauchy morphism, i.e.\ a $\Loc_m$-morphism 
whose image $f(M)$ contains a smooth spacelike Cauchy surface of $M^\prime$. 
Then $f_\ast: \Obs^k(M;\bbZ) \to \Obs^k(M^\prime;\bbZ)$ is an isomorphism. 
\end{propo}
\begin{proof}
Take any smooth spacelike Cauchy surface $\Sigma^\prime \subseteq f(M)$ of $M^\prime$ 
and note that its preimage $\Sigma = f^{-1}(\Sigma^\prime\, )$ is a smooth spacelike Cauchy surface of $M$. 
The diagram of Abelian groups
\begin{flalign}
\xymatrix@C=1.5cm{
\Conf^k_{\rm sc}(M;\bbZ) \ar[r]^-{f_\ast} \ar[d]_-{\iota_\Sigma^\ast \times \iota_\Sigma^\ast}
					&	\Conf_{\rm sc}(M^\prime;\bbZ)
							\ar[d]^-{\iota_{\Sigma^\prime}^\ast \times \iota_{\Sigma^\prime}^\ast}		\\
\dH^{k,m-k}_{\rm c}(\Sigma;\bbZ) \ar[r]_-{f_{\Sigma \ast} \times f_{\Sigma \ast}}
					&	\dH^{k,m-k}_{\rm c}(\Sigma^\prime;\bbZ)
}
\end{flalign}
commutes, its vertical arrows are isomorphisms (cf.\  Corollary \ref{corSCCauchyPbl}) 
and its bottom horizontal arrow is an isomorphism since, by restriction, 
$f$ induces an orientation preserving isometry $f_\Sigma: \Sigma \to \Sigma^\prime$.
Hence $f_\ast$ is an isomorphism and, by using again the natural isomorphism 
$\obs: \Conf_{\rm sc}^k(-;\bbZ) \Rightarrow \Obs^k(-;\bbZ)$, we find 
that $f_\ast: \Obs^k(M;\bbZ) \to \Obs^k(M^\prime;\bbZ)$ is an isomorphism.
\end{proof}

\begin{propo}[Violation of the locality axiom]\label{prpLocalityViolation}
Let $f: M \to M^\prime$ be a morphism in $\Loc_m$.
Then $f_\ast:\Obs^k(M;\bbZ) \to \Obs^k(M^\prime;\bbZ)$ is injective 
if and only if $f_\ast: \H^{k-1,m-k-1}_{\rm sc}(M;\bbT)\to \H^{k-1,m-k-1}_{\rm sc}(M^\prime;\bbT)$ is injective.
For $m=2$ and $k=1$ the latter is always the case, while for $m \geq 3$ and $k \in \{1, \ldots, m-1\}$
there is at least one morphism in $\Loc_m$ violating injectivity. 
\end{propo}
\begin{proof}
Using again the natural isomorphism
$\obs: \Conf^k_{\rm sc}(-;\bbZ) \Rightarrow \Obs^k(-;\bbZ)$,
we can replace $f_\ast:\Obs^k(M;\bbZ) \to \Obs^k(M^\prime;\bbZ)$ in the statement by
$f_\ast:\Conf_{\rm sc}^k(M;\bbZ) \to \Conf_{\rm sc}^k(M^\prime;\bbZ)$.
The proof of the first part follows easily from 
the commutative diagram of short exact sequences
given in Remark \ref{remSequencesSCnatural} and the fact that 
$f_\ast: \Omega^k_{{\rm sc},\bbZ} \cap \ast \Omega^{m-k}_{{\rm sc},\bbZ}(M)\to 
\Omega^k_{{\rm sc},\bbZ} \cap \ast \Omega^{m-k}_{{\rm sc},\bbZ}(M^\prime\, )$
is always injective.
\sk

To prove the second part, we notice that there is a chain of isomorphisms
\begin{flalign}
\H^{k-1,m-k-1}_{\rm sc}(M;\bbT)\simeq \H^{k-1,m-k-1}_{\rm c}(\Sigma;\bbT) \simeq
\H^{m-k,k}(\Sigma;\bbZ)^\star \simeq \H^{m-k,k}(M;\bbZ)^\star~,
\end{flalign}
where $^\star$ denotes Pontryagin duality. The first isomorphism
is from \eqref{eqSCCauchyPblDiagram}, the second is
presented in \cite[Remark~5.7]{BBSS} and
the third simply follows from homotopy invariance of cohomology and $M\simeq \bbR\times \Sigma$.
Hence the counterexamples to injectivity provided
in \cite[Example~6.9]{Becker:2014tla} can be used to prove the present claims.
For the case $m=2$ and $k=1$, see the argument preceding \cite[Proposition~6.11]{Becker:2014tla}.
\end{proof}

The next theorem summarizes the results obtained in this section
in view of the standard axioms of locally covariant field theory \cite{Brunetti:2001dx}. 
In particular, we stress that the locality axiom, 
which requires $f_\ast: \Obs^k(M;\bbZ) \to \Obs^k(M^\prime;\bbZ)$ 
to be injective for all $\Loc_m$-morphisms $f: M \to M^\prime$, does not hold in general. 
\begin{theo}\label{theo:Functorproperties}
The functor $(\Obs^k(-;\bbZ),\tau): \Loc_m \to \PSAb$ 
satisfies the causality and time-slice axioms of locally covariant field theory, 
however the locality axiom is satisfied only in the case $m=2$ and $k=1$, 
while it is violated for $m \geq 3$ and $k \in \{1, \ldots, m-1\}$. 
\end{theo}

%%%%%%%%%%%%%%%%%%%%%%%%%%%%%%%%%%%%%%%%%%%%%%%%
%%%%%%%%%%%%%%%%%%%%%%%%%%%%%%%%%%%%%%%%%%%%%%%%

\section{\label{secQuantization}Quantization}
The quantization of the semi-classical gauge theory $(\Obs^k(-;\bbZ),\tau): \Loc_m \to \PSAb$
can be easily performed by using the well-established techniques
of CCR-algebras, see \cite{Manuceau:1973yn} and also \cite[Appendix A]{Benini:2013ita}
for details. Loosely speaking, given any $m$-dimensional spacetime $M$,
we assign to the presymplectic Abelian group $(\Obs^k(M;\bbZ),\tau)$ the 
$C^\ast$-algebra $\mathfrak{CCR}(\Obs^k(M;\bbZ),\tau)$
that is generated by Weyl symbols $\mathrm{W}(\varphi)$, for all $\varphi\in \Obs^k(M;\bbZ)$,
which satisfy the Weyl relations
\begin{flalign}
\mathrm{W}(\varphi)\,\mathrm{W}(\varphi'\, ) = \exp\big( 2\pi \, {\rm
  i} \, \tau(\varphi,\varphi'\, )\big)\ \mathrm{W}(\varphi+\varphi'\, )~~,\qquad 
\mathrm{W}(\varphi)^\ast = \mathrm{W}(-\varphi)~~,
\end{flalign}
for all $\varphi,\varphi' \in \Obs^k(M;\bbZ)$; by $\exp( 2\pi \, {\rm i} \,(\, \cdot\,)) : \bbT=\bbR/\bbZ \to \bbC$
we denote the embedding of the circle group into the complex numbers.
\sk

More precisely, the CCR-functor $\mathfrak{CCR} : \PSAb \to \CAlg$
from the category of presymplectic Abelian groups to the category of $C^\ast$-algebras
is constructed in detail in \cite[Appendix~A]{Benini:2013ita}. Composing the 
functor $(\Obs^k(-;\bbZ),\tau): \Loc_m \to \PSAb$ with the CCR-functor, we obtain 
a functor from $\Loc_m$ to $\CAlg$ which, according to \cite{Brunetti:2001dx},
should be interpreted as a quantum field theory. We can thereby define a family 
of quantum field theories by setting
\begin{flalign}\label{eqn:QFTfunctors}
\mathfrak{A}^k := \mathfrak{CCR}\circ (\Obs^k(-;\bbZ),\tau) : \mathsf{Loc}_m\longrightarrow \CAlg~,
\end{flalign}
which depend on the degree $k\in \{1,\dots, m-1\}$ of the gauge theory.
The properties of the semi-classical gauge theory from Theorem \ref{theo:Functorproperties}
are preserved by quantization (see e.g.\ the arguments in \cite{Benini:2012vi,Becker:2014tla}), 
which leads us to
\begin{theo}\label{theo:QFTproperties}
The functor $\mathfrak{A}^k : \Loc_m \to \CAlg$ enjoys the following properties:
\begin{itemize}
\item Quantum causality axiom: Let $M_1 \stackrel{f_1}{\longrightarrow} M \stackrel{f_2}{\longleftarrow} M_2$ be
a diagram in $\Loc_m$ such that the images of $f_1$ and $f_2$ are causally disjoint.
Then the subalgebras  ${f_1}_\ast (\mathfrak{A}^k(M_1)) $ and ${f_2}_\ast (\mathfrak{A}^k(M_2)) $
of $\mathfrak{A}^k(M) $ commute.

\item Quantum time-slice axiom: Let $f: M \to M^\prime$ be a Cauchy morphism.
Then $f_\ast: \mathfrak{A}^k(M) \to \mathfrak{A}^k(M^\prime\, )$ is an isomorphism. 

\item Violation of the quantum locality axiom:  Let $f: M \to M^\prime$ be a morphism in $\Loc_m$.
Then $f_\ast:\mathfrak{A}^k(M) \to \mathfrak{A}^k(M^\prime\, )$ is injective 
if and only if $f_\ast: \H^{k-1,m-k-1}_{\rm sc}(M;\bbT)\to \H^{k-1,m-k-1}_{\rm sc}(M^\prime;\bbT)$ is injective.
For $m=2$ and $k=1$ the latter is always the case,
while for $m \geq 3$ and $k \in \{1, \ldots, m-1\}$
there is at least one morphism in $\Loc_m$ violating injectivity. 
\end{itemize}
\end{theo}

%%%%%%%%%%%%%%%%%%%%%%%%%%%%%%%%%%%%%%%%%%%%%%%%
%%%%%%%%%%%%%%%%%%%%%%%%%%%%%%%%%%%%%%%%%%%%%%%%

\section{\label{secAbeliandualities}Quantum duality}
In this section we show that there exist dualities between the quantum field theories
defined in \eqref{eqn:QFTfunctors}. These dualities will be described at the functorial level 
and therefore hold true for all spacetimes $M$ in a coherent (natural) way.
In order to motivate our definition of duality given below,
let us recall that a quantum field theory is a functor $\mathfrak{A} : \Loc_m\to \CAlg$
from the category of $m$-dimensional spacetimes to the category of $C^\ast$-algebras.
The collection of {\em all} $m$-dimensional quantum field theories
is therefore described by the functor category $[\Loc_m,\CAlg]$;
objects in this category are functors $\mathfrak{A} : \Loc_m\to \CAlg$
and morphisms are natural transformations $\eta : \mathfrak{A}\Rightarrow \mathfrak{A}'$.
In physics one calls the functor category $[\Loc_m,\CAlg]$ the ``theory space''
of $m$-dimensional quantum field theories which, being a category,
comes with a natural
notion of equivalence of theories.
\begin{defi}
A duality between two quantum field theories $\mathfrak{A}, \mathfrak{A}' : \Loc_m\to \CAlg$
is a natural isomorphism $\eta : \mathfrak{A}\Rightarrow \mathfrak{A}'$.
\end{defi}

We shall now construct explicit dualities between the quantum field theories
$\mathfrak{A}^k$ and $\mathfrak{A}^{m-k}$ given in \eqref{eqn:QFTfunctors}, for all
$m\geq 2$ and $k \in\{1,\dots,m-1\}$. Our strategy is to define first the dualities
at the level of the semi-classical configuration spaces \eqref{eqConf},
and then lift them to the presymplectic Abelian groups and ultimately to the corresponding 
quantum field theories. For any object $M$ in $\Loc_m$ we define a group homomorphism
\begin{flalign}\label{eqzetamap}
\zeta : \Conf^{m-k}(M;\bbZ)\longrightarrow \Conf^k(M;\bbZ)~, \qquad (h,\tilde h) 
\longmapsto \big(\tilde h, -(-1)^{k\, (m-k)}\, h\big)~,
\end{flalign}
which interchanges (up to a sign) the roles of $h\in \dH^{m-k}(M;\bbZ)$ and $\tilde h\in \dH^{k}(M;\bbZ)$.
We interpret the mapping \eqref{eqzetamap} physically as exchanging the `electric' and `magnetic' sectors 
of the Abelian gauge theory. The map \eqref{eqzetamap} defines a natural isomorphism 
$\zeta: \Conf^{m-k}(-;\bbZ) \Rightarrow \Conf^{k}(-;\bbZ)$ because its components are clearly isomorphisms 
and for any morphism $f : M\to M^\prime$ in $\Loc_m$ the diagram of Abelian groups
\begin{flalign}
\xymatrix{
\ar[d]_-{f^\ast} \Conf^{m-k}(M^\prime;\bbZ) \ar[r]^-{\zeta} & \Conf^{k}(M^\prime;\bbZ) \ar[d]^-{f^\ast}\\
\Conf^{m-k}(M;\bbZ) \ar[r]_-{\zeta} & \Conf^{k}(M;\bbZ) 
}
\end{flalign}
commutes.
\sk

We now dualize \eqref{eqzetamap} with respect to the weakly non-degenerate pairing 
\eqref{eqPairingConf}: Define a group homomorphism $\zeta^\star : \Conf^{k}_{\rm sc}(M;\bbZ)\to \Conf^{m-k}_{\rm sc}(M;\bbZ)$
by the condition
\begin{flalign}
\ips{(h,\tilde h)\, }{\, \zeta^\star(h',\tilde h'\,)} := \ips{\zeta(h,\tilde h)\, }{\, (h',\tilde h'\, )}~,
\end{flalign}
for all $(h',\tilde h'\, )\in \Conf^{k}_{\rm sc}(M;\bbZ)$ and $(h,\tilde h)\in \Conf^{m-k}(M;\bbZ)$.
A quick calculation shows that
\begin{flalign}\label{eqzetaastmap}
\zeta^\star : \Conf^{k}_{\rm sc}(M;\bbZ)\longrightarrow  \Conf^{m-k}_{\rm sc}(M;\bbZ)~, \qquad (h',\tilde h'\, )\longmapsto 
\big( -(-1)^{k\, (m-k)}\, \tilde h',h'\, \big)~.
\end{flalign}
The mapping \eqref{eqzetaastmap} defines a natural isomorphism 
$\zeta^\star : \Conf^{k}_{\rm sc}(-;\bbZ)\Rightarrow \Conf^{m-k}_{\rm sc}(-;\bbZ)$, i.e.\ 
for any morphism $f: M\to M^\prime$ in $\Loc_m$
the diagram of Abelian groups 
\begin{flalign}
\xymatrix{
\ar[d]_-{f_\ast}\Conf^{k}_{\rm sc}(M;\bbZ)\ar[r]^-{\zeta^\star} & \Conf^{m-k}_{\rm sc}(M;\bbZ)\ar[d]^-{f_\ast}\\
\Conf^{k}_{\rm sc}(M^\prime;\bbZ) \ar[r]_-{\zeta^\star} & \Conf^{m-k}_{\rm sc}(M^\prime;\bbZ)
}
\end{flalign}
commutes. We next observe that \eqref{eqzetaastmap} preserves the presymplectic structure
\eqref{eqPreSymplSCConf}: A quick calculation shows that
\begin{flalign}
\sigma\big(\zeta^\star(h,\tilde h)\,,\, \zeta^\star(h',\tilde h'\, )\big) = \sigma\big((h,\tilde h) \,,\, (h', \tilde h'\, )\big)~,
\end{flalign}
for all $(h,\tilde h), (h',\tilde h'\, )\in \Conf^{k}_{\rm sc}(M;\bbZ)$.
Using also the natural isomorphisms $\obs : \Conf^p_{\rm sc}(-;\bbZ) \Rightarrow \Obs^p(-;\bbZ)$ 
given in Proposition \ref{prpSCConfObsIso},
for $p=k$ and $p=m{-}k$,
we find that $\zeta^\star$ defines a natural isomorphism (denoted by the same symbol)
\begin{flalign}
\zeta^\star : \big(\Obs^{k}(-;\bbZ),\tau\big) \Longrightarrow \big(\Obs^{m-k}(-;\bbZ),\tau\big)~
\end{flalign}
between functors from $\Loc_m$ to $\PSAb$.
\sk

We can now state the main result of this section.
\begin{theo}\label{theo:duality}
The $C^\ast$-algebra homomorphism
\begin{flalign}
\eta := \mathfrak{CCR}(\zeta^\star) : \mathfrak{A}^k(M) \longrightarrow \mathfrak{A}^{m-k}(M)
\end{flalign}
defines
a duality between the two quantum field theories $\mathfrak{A}^k,\mathfrak{A}^{m-k} : \Loc_m\to \CAlg$.
\end{theo}
\begin{proof}
We need to show that $\eta$ defines a natural isomorphism 
$\eta : \mathfrak{A}^k\Rightarrow \mathfrak{A}^{m-k}$. Naturality of 
$\eta$ is a direct consequence of naturality of $\zeta^\star$ and the fact that $\mathfrak{CCR}$ is a functor, 
in particular it preserves compositions. As functors preserve isomorphisms it then follows
that $\eta$ is a natural isomorphism.
\end{proof}
\begin{cor}\label{cor:automorphisms}
For $m=2k$ the duality of Theorem \ref{theo:duality} becomes a self-duality, i.e.\ a natural automorphism
$\eta : \mathfrak{A}^k \Rightarrow \mathfrak{A}^{k}$. 
\end{cor}

%%%%%%%%%%%%%%%%%%%%%%%%%%%%%%%%%%%%%%%%%%%%%%%%
%%%%%%%%%%%%%%%%%%%%%%%%%%%%%%%%%%%%%%%%%%%%%%%%

\section{\label{secSelfdual}Self-dual Abelian gauge theory}
In dimension $m=2k$ it makes sense to demand the self-duality condition
\begin{flalign}\label{eqn:selfduality}
\cu\, h = \dcu\, h
\end{flalign}
for a differential character $h\in\dH^k(M;\bbZ)$. Applying the Hodge operator $\ast$
to both sides of \eqref{eqn:selfduality} we obtain
\begin{flalign}
\dcu\, h = -(-1)^{k^2}\, \cu\,h = -(-1)^{k^2}\, \dcu\,h~,
\end{flalign}
which implies that for $k$ even the only solutions to \eqref{eqn:selfduality} are flat fields $h = \kappa(t)$, for 
$t\in \H^{k-1}(M;\bbT)$. In the following we shall focus on the physically much richer and interesting 
case where $k\in 2\bbZ_{\geq0}+1$ is odd.
\sk

The Abelian group of solutions to the self-duality equation \eqref{eqn:selfduality} is denoted
\begin{flalign}
\sdConf^k(M;\bbZ) := \big\{h\in \dH^k(M;\bbZ)\,:\, \cu\,h = \dcu\,h\big\}~.
\end{flalign}
There is a monomorphism 
\begin{flalign}\label{eqn:diagembed}
\mathrm{diag} : \sdConf^k(M;\bbZ) \longrightarrow \Conf^k(M;\bbZ) ~, \qquad h\longmapsto (h,h)
\end{flalign}
to the semi-classical configuration space introduced in \eqref{eqConf}.
Given any smooth spacelike Cauchy surface $\Sigma$ of $M$ with embedding $\iota_{\Sigma} : \Sigma \to M$,
we compose \eqref{eqn:diagembed} with the isomorphism of Theorem \ref{thmCauchyPbl}
and obtain a monomorphism
\begin{flalign}
(\iota^\ast_{\Sigma} \times\iota^\ast_{\Sigma})\circ \mathrm{diag} : \sdConf^k(M;\bbZ) 
\longrightarrow \dH^{k,k}(\Sigma;\bbZ)~, \qquad h\longmapsto (\iota^\ast_{\Sigma} h , \iota^\ast_{\Sigma} h)~,
\end{flalign}
whose image is given by the diagonal in $\dH^{k,k}(\Sigma;\bbZ)$: Given any 
$(h_{\Sigma},h_{\Sigma})\in \dH^{k,k}(\Sigma;\bbZ)$ in the diagonal,
consider the unique solution $(h,\tilde h)\in \Conf^k(M;\bbZ)$
of $\cu\,h =\dcu\,\tilde h$ with initial data $\iota_{\Sigma}^\ast h= h_{\Sigma}$ and 
$\iota^\ast_{\Sigma}\tilde h = h_{\Sigma}$, cf.\ Theorem \ref{thmCauchyPbl}.
Then $(h-\tilde h, \tilde h-h)\in \Conf^k(M;\bbZ)$ satisfy $\iota_\Sigma^\ast(h-\tilde h)=0$
and  $\iota_\Sigma^\ast(\tilde h-h)=0$, hence $\tilde h=h$ by using again Theorem \ref{thmCauchyPbl}.
We have thereby obtained an isomorphism
of Abelian groups
\begin{flalign}\label{eqn:selfdualivp}
\iota^\ast_{\Sigma} : \sdConf^k(M;\bbZ)\longrightarrow \dH^k(\Sigma;\bbZ)~,
\end{flalign}
which we may interpret as in \eqref{eqCauchyPbl} as establishing the well-posedness of the initial value problem
for $h\in\dH^k(M;\bbZ)$ given by 
\begin{flalign}\label{eqn:CauchyPblselfdual}
\cu\,h = \dcu\,h \ ,\qquad \iota^\ast_\Sigma h = h_{\Sigma}~,
\end{flalign}
with initial datum $h_{\Sigma}\in\dH^k(\Sigma;\bbZ)$.
\sk

Similar statements hold true for the Abelian group of solutions of spacelike compact support
to the self-duality equation \eqref{eqn:selfduality}, denoted by
\begin{flalign}
\sdConf^k_{\rm sc}(M;\bbZ) := \big\{h \in\dH^k_{\rm sc}(M;\bbZ)\,:\, \cu \,h = \dcu\,h\big\}~.
\end{flalign}
In particular, there is a monomorphism
\begin{flalign}\label{eqn:scdiagembed}
\mathrm{diag} : \sdConf^k_{\rm sc}(M;\bbZ) \longrightarrow \Conf_{\rm sc}^k(M;\bbZ) ~, \qquad h\longmapsto (h,h)
\end{flalign}
to the Abelian group of semi-classical gauge fields of spacelike compact support introduced in \eqref{eqSCConf}.
Using Corollary \ref{corSCCauchyPbl}, one easily shows that
\begin{flalign}\label{eqn:scselfdualivp}
\iota^\ast_{\Sigma} : \sdConf^k_{\rm sc}(M;\bbZ)\longrightarrow \dH^k_{\rm c}(\Sigma;\bbZ)
\end{flalign}
is an isomorphism, which we may interpret as establishing 
the well-posedness of the initial value problem \eqref{eqn:CauchyPblselfdual}
for $h\in\dH_{\rm sc}^k(M;\bbZ)$ of spacelike compact support 
and initial datum $h_{\Sigma}\in\dH_{\rm c}^k(\Sigma;\bbZ)$ of compact support.
\sk

Similarly to \eqref{eqPairingConf}, there is a weakly non-degenerate $\bbT$-valued pairing
\begin{flalign}
\ips{\cdot}{\cdot}_{\mathfrak{sd}} : \sdConf^k(M;\bbZ)\times \sdConf^k_{\rm sc}(M;\bbZ) 
\longrightarrow \bbT~, \qquad (h,h'\, )\longmapsto
\ips{\iota^\ast_{\Sigma} h}{\iota^\ast_{\Sigma}h'\, }_{\rm c}^{}~,
\end{flalign}
which is independent of the choice of Cauchy surface $\Sigma$ of $M$.\footnote{This is 
demonstrated by a proof similar to that of Lemma \ref{lemPairingSigmaindep}.}
Thus there is a monomorphism
\begin{flalign}\label{eqn:selfdualobsiso}
\obs_{\mathfrak{sd}}  :  \sdConf^k_{\rm sc}(M;\bbZ)\longrightarrow \sdConf^k(M;\bbZ)^\star ~, \qquad 
h'\longmapsto \ips{\cdot}{h'\, }_{\mathfrak{sd}}
\end{flalign}
to the character group of $\sdConf^k(M;\bbZ)$, 
whose image is denoted $\sdObs^k(M;\bbZ)$ and called
the Abelian group of semi-classical observables on $\sdConf^k(M;\bbZ)$.
\sk

Analogously to Proposition \ref{prpPreSympl} we define a $\bbT$-valued presymplectic structure
\begin{flalign}
\tau_{\mathfrak{sd}} : \sdObs^k(M;\bbZ)\times \sdObs^k(M;\bbZ)\longrightarrow \bbT~, \qquad
(\varphi,\varphi'\, )\longmapsto \ips{I(\obs_{\mathfrak{sd}}^{-1}\varphi)}{\obs_{\mathfrak{sd}}^{-1}\varphi'\, }_{\mathfrak{sd}}~.
\end{flalign}
Up to the isomorphism $\obs_{\mathfrak{sd}}^{-1} : \sdObs^k(M;\bbZ)\to  \sdConf^k_{\rm sc}(M;\bbZ)$ 
induced by \eqref{eqn:selfdualobsiso}, the presymplectic structure reads as
\begin{flalign}\label{eqn:sigmasd}
\sigma_{\mathfrak{sd}} :  \sdConf^k_{\rm sc}(M;\bbZ)\times  \sdConf^k_{\rm sc}(M;\bbZ) \longrightarrow \bbT~, \qquad
(h,h'\, ) \longmapsto \ips{I(h)}{h'\, }_{\mathfrak{sd}}~.
\end{flalign}
The radical of $\sigma_{\mathfrak{sd}}$ coincides with the kernel
of $I : \sdConf^k_{\rm sc}(M;\bbZ) \to \sdConf^k(M;\bbZ) $, hence the radical of $\tau_{\mathfrak{sd}} $
is $\obs_{\mathfrak{sd}}(\mathrm{ker} \,I)$.
\sk

Using arguments similar to those of Section \ref{secObs}, one can show that the presymplectic Abelian groups
$(\sdObs^k(M;\bbZ),\tau_{\mathfrak{sd}})$ for the self-dual field theory are functorial, i.e.\
we have constructed a functor 
\begin{flalign}
(\sdObs^k(-;\bbZ),\tau_{\mathfrak{sd}}) : \Loc_{2k} \longrightarrow \PSAb~.
\end{flalign}
Composing with the CCR-functor from Section \ref{secQuantization} we obtain quantum field theories
\begin{flalign}
\mathfrak{sdA}^k := \mathfrak{CCR}\circ (\sdObs^k(-;\bbZ),\tau_{\mathfrak{sd}}) : \Loc_{2k}\longrightarrow \CAlg~,
\end{flalign}
for all $k\in 2\bbZ_{\geq0} +1$, which quantize the self-duality equation \eqref{eqn:selfduality}.
Using similar arguments as those of Section \ref{subPropPSAb},
one can show that these quantum field theories satisfy the same
properties as those listed in
Theorem \ref{theo:QFTproperties}.
\begin{theo}\label{theo:selfdualQFTproperties}
The functor $\mathfrak{sdA}^k  : \Loc_{2k} \to \CAlg$ enjoys the following properties:
\begin{itemize}
\item Quantum causality axiom: Let $M_1 \stackrel{f_1}{\longrightarrow} M \stackrel{f_2}{\longleftarrow} M_2$ be
a diagram in $\Loc_{2k}$ such that the images of $f_1$ and $f_2$ are causally disjoint.
Then the subalgebras  ${f_1}_\ast (\mathfrak{sdA}^k(M_1)) $ and ${f_2}_\ast (\mathfrak{sdA}^k(M_2)) $
of $\mathfrak{sdA}^k(M) $ commute.

\item Quantum time-slice axiom: Let $f: M \to M^\prime$ be a Cauchy morphism.
Then $f_\ast: \mathfrak{sdA}^k(M) \to \mathfrak{sdA}^k(M^\prime\, )$ is an isomorphism. 

\item Violation of the quantum locality axiom: Let $f: M \to M^\prime$ be a morphism in $\Loc_{2k}$.
Then $f_\ast:\mathfrak{sdA}^k(M) \to \mathfrak{sdA}^k(M^\prime\, )$ is injective 
if and only if $f_\ast: \H^{k-1}_{\rm sc}(M;\bbT)\to \H^{k-1}_{\rm sc}(M^\prime;\bbT)$ is injective.
For $k=1$ the latter is always the case, while for $k \in 2\bbZ_{\geq0}+3 $
there is at least one morphism in $\Loc_{2k}$ violating injectivity. 
\end{itemize}
\end{theo}
\begin{rem}
We address the question how the self-dual quantum field theories $\mathfrak{sdA}^k$, which quantize
the self-duality equation \eqref{eqn:selfduality}, are related to the self-dualities of
the quantum field theories $\mathfrak{A}^k$ established in Corollary \ref{cor:automorphisms}. 
Let $k\in 2\bbZ_{\geq0}+1$ and consider any object $M$ in $\Loc_{2k}$. 
The self-duality  \eqref{eqzetaastmap} on $\Conf^k_{\rm sc}(M;\bbZ)$ then reduces
to $\zeta^\star(h',\tilde h'\, ) = (\tilde h',h'\, )$, i.e.\ it simply interchanges $h'$ and $\tilde h'$.
The Abelian group of invariants under this self-duality is given by the diagonal
\begin{flalign}
\Conf^k_{\rm sc}(M;\bbZ)^{\mathrm{inv}} := 
\big\{(h',\tilde h'\,)\in \Conf^k_{\rm sc}(M;\bbZ) \, :\, \zeta^\star(h',\tilde h'\, ) = (h',\tilde h'\, ) \big\}
=\big\{(h',h'\, )\in \Conf^k_{\rm sc}(M;\bbZ)\big\}~,
\end{flalign}
which by \eqref{eqn:scdiagembed} is isomorphic to $\sdConf^k_{\rm sc}(M;\bbZ)$.
Restricting the presymplectic structure \eqref{eqPreSymplSCConf} to 
the invariants $\Conf^k_{\rm sc}(M;\bbZ)^{\mathrm{inv}}$ then yields
\begin{flalign}
\sigma : \Conf^k_{\rm sc}(M;\bbZ)^{\mathrm{inv}}\times \Conf^k_{\rm sc}(M;\bbZ)^{\mathrm{inv}}\longrightarrow \bbT~, \qquad
\big((h,h)\,,\,(h',h'\, )\big)\longmapsto 2\,\sigma_{\mathfrak{sd}}(h,h'\, )~,
\end{flalign}
where $\sigma_{\mathfrak{sd}}$ is the presymplectic
structure on $\sdConf^k_{\rm sc}(M;\bbZ)$ given in \eqref{eqn:sigmasd}. 
Due to the prefactor~$2$, it follows that $(\sdConf^k_{\rm sc}(M;\bbZ),\sigma_{\mathfrak{sd}})$
and $(\Conf^k_{\rm sc}(M;\bbZ)^{\mathrm{inv}},\sigma)$ are {\em not} 
isomorphic as presymplectic Abelian groups, but only as Abelian groups.
Moreover, the $C^\ast$-algebras $\mathfrak{sdA}^k(M)$ and $\mathfrak{A}^k(M)^{\mathrm{inv}}$
(i.e.\ the $C^\ast$-subalgebra of $\mathfrak{A}^k(M)$ which is generated by the invariant Weyl symbols
$\mathrm{W}(\obs(h',h'\, ))$, for all $(h',h'\, )\in\Conf^k_{\rm sc}(M;\bbZ)^{\mathrm{inv}} $) 
are in general not isomorphic. 
Thus even though the quantum field theories $\mathfrak{sdA}^k : \Loc_{2k}\to \CAlg$
and $\mathfrak{A}^k(-)^{\mathrm{inv}} : \Loc_{2k}\to \CAlg$ are similar, 
they are strictly speaking not isomorphic. In particular, due to effects which are 
caused by $\bbZ_2$-torsion elements in the cohomology groups $\H^k(M;\bbZ)$, the latter theory typically
has a bigger center than the former theory. An explicit example of this fact is illustrated below. 
\end{rem}

\begin{ex}
Fix any $k \in 2\bbZ_{\geq 0}+3$ and consider the lens space $L = \mathbb{S}^{2k-3}/\bbZ_2$ 
obtained as the quotient of the $2k{-}3$-sphere $\mathbb{S}^{2k-3}$ by the antipodal $\bbZ_2$-action.
Take any object $M$ in $\Loc_{2k}$ which admits a smooth spacelike Cauchy surface $\Sigma$ diffeomorphic
to $\bbT^2 \times L$, where $\bbT^2$ is the $2$-torus. Since the Cauchy surface
$\Sigma$ is compact, the notion of spacelike compact support becomes irrelevant for this spacetime
$M$ and in particular the homomorphism $I : \sdConf^k_{\rm sc}(M;\bbZ) \to \sdConf^k(M;\bbZ) $ 
reduces to the identity. Using standard results on the homology groups of lens spaces, 
see e.g.\ \cite[Chapter 2, Example 2.43]{Hatcher}, and the universal coefficient theorem 
for cohomology, one shows that $\H^{k-1}(L;\bbZ) \simeq \bbZ_2$. Using the K\"unneth theorem
we find that $\H^k(\Sigma;\bbZ)$ has a direct summand $(\bbZ_2)^2$. 
In particular, there exists $t \in \H^k(\Sigma;\bbZ)$ such that $t \neq 0$, but $2 t = 0$. 
Recalling that $\ch$ is surjective (cf.\ \eqref{eqDiffCohoDiagram}), 
we find $f \in \dH^k(\Sigma;\bbZ)$ such that $\ch\, f = t$. 
It follows that there exists $A \in \Omega^{k-1}(\Sigma)$ such that $\iota[A] = 2 f$. 
Introducing $h_\Sigma = f - \iota[A/2] \in \dH^k(\Sigma;\bbZ)$, 
by construction we obtain $h_\Sigma \neq 0$ (otherwise $t$ would be trivial) and $2 h_\Sigma = 0$. 
Solving the initial value problem \eqref{eqn:CauchyPblselfdual} provides $h \in \mathfrak{sdC}^k(M;\bbZ)$ 
with $h \neq 0$, but $2 h = 0$. 
In fact, $2 h \in \mathfrak{sdC}^k(M;\bbZ)$ solves \eqref{eqn:CauchyPblselfdual} 
with initial datum $2 h_\Sigma = 0$. 
Since $\Sigma$ is compact, the presymplectic structure \eqref{eqn:sigmasd} is weakly non-degenerate. 
In particular, being non-zero, $h \in \mathfrak{sdC}^k(M)$ is not in the radical. 
Conversely, taking into account also $(h,h) \in \Conf^k(M;\bbZ)^\mathrm{inv}$, 
we find $\sigma((h,h),(h',h'\, )) = 2 \sigma_\mathfrak{sd}(h,h'\, ) =
\sigma_\mathfrak{sd}(2 h,h'\, ) = 0$ 
for all $(h',h'\, ) \in \Conf^k(M)^{\mathrm{inv}}$. This shows that the center of $\mathfrak{A}^k(M)^{\mathrm{inv}}$
is bigger than that of $\mathfrak{sdA}^k (M)$ for this particular spacetime $M$.
\end{ex}

%%%%%%%%%%%%%%%%%%%%%%%%%%%%%%%%%%%%%%%%%%%%%%%%
%%%%%%%%%%%%%%%%%%%%%%%%%%%%%%%%%%%%%%%%%%%%%%%%

\section*{Acknowledgments}
We thank Ulrich Bunke for helpful discussions and the anonymous referees for their comments and remarks.
This work was supported in part by the Action MP1405 QSPACE from the 
European Cooperation in Science and Technology (COST).  
The work of C.B.\ is partially supported by the Collaborative Research Center (SFB) 
``Raum Zeit Materie'', funded by the Deutsche Forschungsgemeinschaft (DFG, Germany).
The work of M.B.\ is supported partly by a Research Fellowship of the Della Riccia Foundation (Italy) 
and partly by a Postdoctoral Fellowship of the Alexander von Humboldt Foundation (Germany). 
The work of A.S.\ is supported by a Research Fellowship of the Deutsche Forschungsgemeinschaft (DFG, Germany). 
The work of R.J.S.\ is partially supported by the Consolidated Grant ST/L000334/1 
from the UK Science and Technology Facilities Council.

%%%%%%%%%%%%%%%%%%%%%%%%%%%%%%%%%%%%%%%%%%%%%%%%
%%%%%%%%%%%%%%%%%%%%%%%%%%%%%%%%%%%%%%%%%%%%%%%%

\appendix

\section{\label{secTechLemmas}Technical lemmas}
In this appendix we prove five lemmas which are used in the main text.
Some of these proofs are rather technical and also make use of results in the companion 
paper \cite{BBSS}, in which case we give precise references.

\begin{lem}\label{lemChainHomotopy}
Let $M$ be a time-oriented $m$-dimensional globally hyperbolic Lorentzian manifold
and $\Sigma$ a smooth spacelike Cauchy surface of $M$. 
Consider the embedding $\iota_{\Sigma} : \Sigma \to M$ of $\Sigma$ into $M$,
and the projection $\pi_{\Sigma} : M\to \Sigma$ of $M$ onto $\Sigma$ which is induced by a
choice of diffeomorphism $M\simeq \bbR\times \Sigma$ such that $\iota_{\Sigma}(\Sigma) \simeq \{0\}\times \Sigma$,
cf.\ \cite{Bernal:2004gm,Bernal:2005qf}. Then:
\begin{itemize}
\item[(i)] Inducing $\iota_{\Sigma}$ and $\pi_{\Sigma}$ to smooth singular chains, i.e.\ 
$\iota_{\Sigma\,\ast} : C_{\sharp}(\Sigma) \to C_{\sharp}(M)$ and $\pi_{\Sigma\,\ast} : C_{\sharp}(M)\to C_{\sharp}(\Sigma)$,
we have $\pi_{\Sigma\,\ast}\circ \iota_{\Sigma\,\ast} = \id$
and $\iota_{\Sigma\,\ast}\circ \pi_{\Sigma\,\ast} - \id = \del\circ  h_{\Sigma} + h_{\Sigma}\circ\del$,
for a chain homotopy $h_{\Sigma} : C_{\sharp} (M)\to C_{\sharp+1}(M)$. In particular, $\iota_{\Sigma}$ and $\pi_{\Sigma}$
induce isomorphisms on smooth singular homology:
\begin{flalign}
\xymatrix{
\H_\sharp(M) \ar@<3pt>[r]^-{\pi_{\Sigma\,\ast}}		&	\H_\sharp(\Sigma) \ar@<1pt>[l]^-{\iota_{\Sigma\,\ast}}~.
}
\end{flalign}
\item[(ii)] Let $G$ be an Abelian group. Inducing $\iota_{\Sigma}$ and $\pi_{\Sigma}$ to smooth singular $G$-valued cochains, 
i.e.\ $\iota_{\Sigma}^\ast : C^{\sharp}(M;G) \to C^{\sharp}(\Sigma;G)$ and 
$\pi_{\Sigma}^{\ast} : C^{\sharp}(\Sigma;G)\to C^{\sharp}(M;G)$,
we have $\iota_{\Sigma}^\ast \circ \pi_{\Sigma}^\ast = \id$
and $\pi_{\Sigma}^\ast\circ \iota_{\Sigma}^\ast - \id = \cdel\circ  h_{\Sigma}^\star + h_{\Sigma}^\star\circ\cdel$,
for a cochain homotopy $h_{\Sigma}^\star : C^{\sharp} (M;G)\to C^{\sharp-1}(M;G)$. 
In particular, $\iota_{\Sigma}$ and $\pi_{\Sigma}$
induce isomorphisms on smooth singular cohomology with coefficients in~$G$:
\begin{flalign}
\xymatrix{
\H^\sharp(M;G) \ar@<-1pt>[r]_-{\iota_\Sigma^\ast}	&	\H^\sharp(\Sigma;G) \ar@<-3pt>[l]_-{\pi_\Sigma^\ast}~.
}
\end{flalign}
\end{itemize}
\end{lem}
\begin{proof}
We shall denote points by $x\in\Sigma$ and $(t,x)\in M\simeq \bbR\times \Sigma$.
By construction we have $\pi_{\Sigma}\circ \iota_\Sigma = \id_{\Sigma}$.
Notice further that $\iota_{\Sigma}\circ \pi_{\Sigma}$ and the identity $\id_M$
are homotopic via
\begin{flalign}\label{eqn:tmpHomotopy}
H_{\Sigma} : [0,1]\times M\longrightarrow M~, \qquad \big(s\,,\,(t,x) \big)\longmapsto (s\,t,x)~.
\end{flalign}
As usual, see for example the proof of \cite[Theorem 2.10]{Hatcher}, this homotopy induces
the desired chain homotopy $h_{\Sigma} : C_{\sharp} (M)\to C_{\sharp+1}(M)$, which proves item (i).
Item (ii) then follows by defining $h_\Sigma^\star = \Hom(h_\Sigma,G): C^{\sharp} (M;G)\to C^{\sharp-1}(M;G)$.
\end{proof}

\begin{lem}\label{lemRelChainHomotopy}
Under the same hypotheses as in Lemma \ref{lemChainHomotopy}, let $K\subseteq \Sigma$
be a compact subset. Then:
\begin{itemize}
\item[(i)] Inducing $\iota_{\Sigma}$ and $\pi_{\Sigma}$ to relative smooth singular chains, i.e.\ 
$\iota_{\Sigma\,\ast} : C_{\sharp}(\Sigma,\Sigma\setminus K) \to C_{\sharp}(M,M\setminus J(K))$ 
and $\pi_{\Sigma\,\ast} : C_{\sharp}(M,M\setminus J(K))\to C_{\sharp}(\Sigma,\Sigma\setminus K)$,
we have $\pi_{\Sigma\,\ast}\circ \iota_{\Sigma\,\ast} = \id$
and $\iota_{\Sigma\,\ast}\circ \pi_{\Sigma\,\ast} - \id = 
\del\circ  h_{\Sigma} + h_{\Sigma}\circ\del$, for a chain homotopy $h_{\Sigma} : C_{\sharp} (M,M\setminus J(K))
\to C_{\sharp+1}(M,M\setminus J(K))$. In particular, $\iota_{\Sigma}$ and $\pi_{\Sigma}$
induce an isomorphism on relative smooth singular homology:
\begin{flalign}
\xymatrix{
\H_\sharp(M,M\setminus J(K)) \ar@<3pt>[r]^-{\pi_{\Sigma\,\ast}}		&	
\H_\sharp(\Sigma,\Sigma\setminus K) \ar@<1pt>[l]^-{\iota_{\Sigma\,\ast}}~~.
}
\end{flalign}
\item[(ii)] Let $G$ be an Abelian group. Inducing $\iota_{\Sigma}$ and $\pi_{\Sigma}$ to relative 
smooth singular $G$-valued cochains, 
i.e.\ $\iota_{\Sigma}^\ast : C^{\sharp}(M,M\setminus J(K);G) \to C^{\sharp}(\Sigma,\Sigma\setminus K;G)$ and 
$\pi_{\Sigma}^{\ast} : C^{\sharp}(\Sigma,\Sigma\setminus K ;G)\to C^{\sharp}(M,M\setminus J(K);G)$,
we have $\iota_{\Sigma}^\ast \circ \pi_{\Sigma}^\ast = \id$
and $\pi_{\Sigma}^\ast\circ \iota_{\Sigma}^\ast - \id
= \cdel\circ  h_{\Sigma}^\star + h_{\Sigma}^\star\circ\cdel$,
for a cochain homotopy $h_{\Sigma}^\star : C^{\sharp} (M,M\setminus J(K);G)\to C^{\sharp-1}(M,M\setminus J(K);G)$. 
In particular, $\iota_{\Sigma}$ and $\pi_{\Sigma}$
induce an isomorphism on relative smooth singular cohomology with coefficients in $G$:
\begin{flalign}
\xymatrix{
\H^\sharp(M,M\setminus J(K);G) \ar@<-1pt>[r]_-{\iota_\Sigma^\ast}	&	
\H^\sharp(\Sigma,\Sigma\setminus K;G) \ar@<-3pt>[l]_-{\pi_\Sigma^\ast}~~.
}
\end{flalign}
\end{itemize}
\end{lem}
\begin{proof}
Notice that $\iota_{\Sigma} : \Sigma \to M$ maps $\Sigma\setminus K$ to $M\setminus J(K)$ because
$\Sigma$ is by assumption spacelike.
Moreover, $\pi_{\Sigma} : M \to \Sigma$ maps $M\setminus J(K)$ to $\Sigma\setminus K$: 
Assume there exists $(t,x)\in M\setminus J(K)$ such that $\pi_{\Sigma}(t,x) = x\in K$;
then the curve $\gamma : [0,1] \to M\,,~s\mapsto (s\,t,x)$ connecting $(0,x)$ with $(t,x)$
is timelike, hence $(t,x)\in J(K)$ which is a contradiction.
Similarly, the homotopy \eqref{eqn:tmpHomotopy} restricts
to $H_{\Sigma} : [0,1]\times \big(M\setminus J(K)\big) \to M\setminus J(K)$.
The rest of the proof then follows that of Lemma \ref{lemChainHomotopy}.
\end{proof}

\begin{lem}\label{lemSCPushforward}
Let $f: M \to M^\prime$ be a morphism in $\Loc_m$ 
and denote by $\mathcal{K}_M$ the directed set of compact subsets of $M$. 
Consider the natural transformation
\begin{flalign}\label{eqPullbackRelConf}
f^\ast: \Conf^k(M^\prime,M^\prime \setminus J(f(-));\bbZ) \Longrightarrow \Conf^k(M,M \setminus J(-);\bbZ)~
\end{flalign}
between functors from $\mathcal{K}_M$ to $\Ab$ induced by $f$.
Then, for each smooth spacelike Cauchy surface $\Sigma$ of $M$, 
the restriction of $f^\ast$ to the directed set $\mathcal{K}_\Sigma \subseteq \mathcal{K}_M$ 
of compact subsets of $\Sigma$ is a natural isomorphism. 
In particular, we can consider the natural transformation 
\begin{flalign}
(f^\ast)^{-1}: \Conf^k(M,M \setminus J(-);\bbZ) \Longrightarrow  \Conf^k(M^\prime,M^\prime \setminus J(f(-));\bbZ)~
\end{flalign}
between functors from $\mathcal{K}_\Sigma$ to $\Ab$.
Then the pushforward for semi-classical gauge fields of spacelike compact support
\begin{equation}\label{eqPushforwardSCConf}
f_\ast: \Conf^k_{\rm sc}(M;\bbZ) \longrightarrow \Conf^k_{\rm sc}(M^\prime;\bbZ)
\end{equation}
is canonically induced by the colimit prescription in \eqref{eqSCConf} restricted to $\mathcal{K}_\Sigma$, 
see also Remark~\ref{remCofinal}, and by the universal property of colimits. 
\end{lem}

\begin{proof}
For each $K \subseteq M$ compact, we note that $f: M \to M^\prime$ in $\Loc_m$ induces 
an open embedding $f:(M,M \setminus J(K)) \to (M^\prime,M^\prime \setminus J(f(K)))$ of 
pairs which is compatible with the inclusions of the given submanifolds. 
Looking at an inclusion $K \subseteq K^\prime$ of compact subsets of $M$, 
one realizes that both $(M,M \setminus J(-))$ and $(M^\prime,M^\prime \setminus J(f(-)))$ 
are functors from $\mathcal{K}_M$ to $\Pair^\op$, the opposite category 
of the category $\Pair$ of pairs of manifolds 
with submanifold preserving smooth maps as morphisms;
moreover, $f:(M,M \setminus J(-)) \Rightarrow (M^\prime,M^\prime \setminus J(f(-)))$ 
is a natural transformation between these functors. 
Therefore, applying the functor $\dH^{k,m-k}(-;\bbZ): \Pair^\op \to \Ab$, 
cf.\ \cite[Section~3.1]{BBSS}, 
we obtain the pullback along $f$ as a natural transformation
\begin{equation}
f^\ast: \dH^{k,m-k}(M^\prime,M^\prime \setminus J(f(-));\bbZ) \Longrightarrow \dH^{k,m-k}(M,M \setminus J(-);\bbZ)~
\end{equation}
between functors from $\mathcal{K}_M$ to $\Ab$.
Since $f$ is a morphism in $\Loc_m$, hence in particular an isometry, 
and $\cu$ is a natural transformation for relative differential characters, 
see \cite[Section~3.1]{BBSS}, 
$f^\ast$ maps relative semi-classical gauge fields on $M^\prime$ to relative
semi-classical gauge fields on $M$, 
so that we obtain the natural transformation displayed in \eqref{eqPullbackRelConf}. 
\sk

We will now show that the restriction to $\mathcal{K}_\Sigma$ 
of the natural transformation \eqref{eqPullbackRelConf} is a natural isomorphism. 
For each $K \subseteq \Sigma$ compact, we choose an open neighborhood $U \subseteq \Sigma$ of $K$ 
with compact closure $\overline U$. We denote by $j: U \to \Sigma$ the open embedding 
induced by the inclusion. Observing that $f(\, \overline U\, ) \subseteq M^\prime$ 
is a spacelike and acausal compact submanifold with boundary of $M^\prime$, 
by~\cite[Theorem 1.1]{Bernal:2005qf} there is a smooth spacelike Cauchy surface 
$\Sigma^\prime$ of $M^\prime$ extending $f(\, \overline U\, )$
whose embedding in $M^\prime$ is denoted by $\iota_{\Sigma^\prime}: \Sigma^\prime \to M^\prime$. 
We also denote by $f_U: U \to \Sigma^\prime$ the open embedding induced by the restriction of $f$. 
By construction, the diagram
\begin{flalign}\label{eqAux3}
\xymatrix{
(U,U \setminus K) \ar[rr]^-{f_U} \ar[d]_-j			&	&	(\Sigma^\prime,\Sigma^\prime \setminus f(K))
															\ar[dd]^-{\iota_{\Sigma^\prime}}			\\
(\Sigma,\Sigma \setminus K) \ar[d]_-{\iota_\Sigma}														\\
(M,M \setminus J(K)) \ar[rr]_-f						&&	(M^\prime,M^\prime \setminus J(f(K)))
}
\end{flalign}
in the category  $\Pair$ commutes.
Therefore, applying the functor $\dH^{k,m-k}(-;\bbZ): \Pair^\op \to \Ab$ 
and recalling that the pullback along $f$ maps relative semi-classical gauge fields to relative semi-classical gauge fields, 
we obtain a new commutative diagram
\begin{flalign}\label{eqAux4}
\xymatrix{
\Conf^k(M^\prime,M^\prime \setminus J(f(K));\bbZ) \ar[rr]^-{f^\ast} \ar[dd]_-{\iota_{\Sigma^\prime}^\ast  }
		&&	\Conf^k(M,M \setminus J(K);\bbZ) \ar[d]^-{\iota_\Sigma^\ast}								\\
		&&	\dH^{k,m-k}(\Sigma,\Sigma \setminus K;\bbZ) \ar[d]^-{j^\ast }								\\
\dH^{k,m-k}(\Sigma^\prime,\Sigma^\prime \setminus f(K);\bbZ) \ar[rr]	_{f_U^\ast }
		&&	\dH^{k,m-k}(U,U \setminus K;\bbZ)
}
\end{flalign}
in the category of Abelian groups $\Ab$.
Using Theorem \ref{thmRelCauchyPbl} and the excision theorem \cite[Theorem~3.8]{BBSS}
we find that the vertical and bottom horizontal arrows are isomorphisms. 
In particular, note that $f_U$ can be factored as the composition of a diffeomorphism onto its image 
followed by the inclusion of its image into the target, hence by excision
$f_U^\ast: \dH^k(\Sigma^\prime,\Sigma^\prime \setminus f(K);\bbZ) \to \dH^k(U,U \setminus K;\bbZ)$ 
is an isomorphism. It follows that the top horizontal arrow is an isomorphism too.
\end{proof}

The proofs of our final two lemmas will rely extensively on Lemma~\ref{lemSCPushforward}. 
In particular, we will adopt the following approach. 
Starting from a semi-classical gauge field of  spacelike compact support
and unraveling the directed colimit in \eqref{eqSCConf}, we will represent it by a gauge field 
relative to the complement of $J(K)$ for a suitable compact subset $K$ of a smooth spacelike Cauchy surface. 
Then we use Lemma \ref{lemSCPushforward} to represent the pushforward 
of the given spacelike compact gauge field by the image under the inverse of the pullback 
of the corresponding relative gauge field.

\begin{lem}\label{lemPairingSigmaindep}
The pairing \eqref{eqPairingConf} does not depend on the choice of Cauchy surface $\Sigma$.
Moreover, for any morphism $f: M \to M^\prime$ in $\Loc_m$ the diagram of Abelian groups
\begin{flalign}\label{eqPairingConfNat}
\xymatrix{
\ar[d]_-{\id\times f_\ast}\Conf^k(M^\prime;\bbZ) \times \Conf^k_{\rm sc}(M;\bbZ) \ar[rr]^-{f^\ast \times \id} && 
\Conf^k(M;\bbZ) \times \Conf^k_{\rm sc}(M;\bbZ)\ar[d]^-{\ips{\cdot}{\cdot}}\\
\Conf^k(M^\prime;\bbZ) \times \Conf^k_{\rm sc}(M^\prime;\bbZ)\ar[rr]_-{\ips{\cdot}{\cdot}} && \bbT
}
\end{flalign}
commutes.
\end{lem}
\begin{proof}
We first prove independence of the pairing \eqref{eqPairingConf} on the choice of Cauchy surface $\Sigma$ 
used to evaluate it. For this, we choose any two smooth spacelike Cauchy surfaces $\Sigma$ 
and $\Sigma^\prime$ of $M$. Let $(h,\tilde h) \in \Conf^k(M;\bbZ)$, $K \subseteq \Sigma$ compact 
and $(h',\tilde h'\, ) \in \Conf^k(M,M \setminus J(K);\bbZ)$. 
Note that $K^\prime = \Sigma^\prime \cap J(K)$ is compact and that
$J(K)\subseteq J(K^\prime\, )$. 
Let $\mu$ denote the unique element of $\H_{m-1}(\Sigma,\Sigma \setminus K)$ 
which agrees with the orientation of $\Sigma$ for each point of $K$. 
Similarly, let $\mu^\prime$ denote the unique element of 
$\H_{m-1}(\Sigma^\prime,\Sigma^\prime \setminus K^\prime\, )$ 
which agrees with the orientation of $\Sigma^\prime$ at each point of $K^\prime$. 
By means of the isomorphisms in Lemma \ref{lemRelChainHomotopy}, we can compare $\mu$ 
with $\tilde \mu = \pi_{\Sigma\,\ast} \, \iota_{\Sigma^\prime\,\ast} \mu^\prime \in \H_{m-1}(\Sigma,\Sigma \setminus K)$. 
Since the orientations of both $\Sigma^\prime$ and $\Sigma$ 
are chosen consistently with the orientation and time-orientation of $M$, 
for each point of $K$, $\tilde \mu$ agrees with the orientation of $\Sigma$. 
In particular \cite[Lemma 3.27]{Hatcher} entails that 
$\tilde \mu = \mu \in \H_{m-1}(\Sigma,\Sigma \setminus K)$, therefore 
$\iota_{\Sigma^\prime\,\ast} \mu^\prime = \iota_{\Sigma\,\ast} \mu \in \H_{m-1}(M,M \setminus J(K))$ 
by Lemma \ref{lemRelChainHomotopy}. Let $\nu \in Z_{m-1}(\Sigma,\Sigma \setminus K)$ 
and $\nu^\prime \in Z_{m-1}(\Sigma^\prime,\Sigma^\prime \setminus
K^\prime\, )$ be cycles 
representing $\mu$ and $\mu^\prime$ respectively. Hence we obtain $\gamma \in C_m(M,M \setminus J(K))$
such that $\iota_{\Sigma \ast} \nu - \iota_{\Sigma^\prime \ast} \nu^\prime = \del \gamma$. 
Taking into account also \cite[Definition 3.1]{BBSS}, we get 
\begin{multline}
\ip{\iota_{\Sigma }^\ast (h,\tilde h)\, }{\, \iota_{\Sigma }^\ast( h',\tilde h'\, )}_\Sigma
	- \ip{\iota_{\Sigma' }^\ast( h,\tilde h)\, }{\, \iota_{\Sigma' }^\ast( h',\tilde h'\, )}_{\Sigma^\prime} 
= \big(\, \tilde h \cdot h' - (-1)^{k\, (m-k)} \, h \cdot \tilde h'\,\big)(\del \gamma)	\\[4pt]
	= \int_\gamma \, \cu\big(\tilde h \cdot h' - (-1)^{k\, (m-k)}\, h \cdot \tilde h'\, \big) \mod \bbZ~,
\end{multline}	
where the subscripts $_\Sigma$ and $_{\Sigma^\prime}$ denote the Cauchy surfaces 
which have been used to evaluate the pairing \eqref{eqPairingConf}. 
Furthermore, using \cite[eq.\ (3.31)]{BBSS} 
together with the identities $\cu\, h = \ast \, \cu\, \tilde h$ and $\cu\, h' = \ast \, \cu\, \tilde h'$, one has 
\begin{flalign}
\cu\,\big(\tilde h \cdot h' - (-1)^{k\, (m-k)} \, h \cdot \tilde h'\, \big) 
= \cu\, \tilde h \wedge \ast \, \cu\, \tilde h' - \cu\, \tilde h'
\wedge \ast \, \cu\, \tilde h = 0~.
\end{flalign}

To complete the proof, we show that the diagram \eqref{eqPairingConfNat} commutes. 
Let $(h,\tilde h) \in \Conf^k(M^\prime;\bbZ)$ and $(h',\tilde h'\, ) \in \Conf^k_{\rm sc}(M;\bbZ)$. 
Consider a smooth spacelike Cauchy surface $\Sigma$ of $M$ and $K \subseteq \Sigma$ compact 
such that $(h',\tilde h'\, ) \in \Conf^k(M,M \setminus J(K);\bbZ)$. 
Since $f$ is a morphism in the category $\Loc_m$, the induced map 
$f: (M,M \setminus J(K)) \to (M^\prime,M^\prime \setminus J(f(K)))$ is a morphism in $\Pair$, 
in particular an open embedding between the $m$-manifolds $M$ and $M^\prime$ 
which maps the open subset $M \setminus J(K)$ to $M' \setminus f(J(K))$. We 
represent $f_\ast (h',\tilde h'\, ) \in \Conf^k_{\rm sc}(M^\prime;\bbZ)$ 
by $(f^\ast)^{-1} (h',\tilde h'\, ) \in \Conf^k(M^\prime,M^\prime \setminus J(f(K));\bbZ)$, 
and interpret $(h',\tilde h'\, )$ as its representative in $\Conf^k(M,M \setminus J(K);\bbZ)$.
Recalling the proof of Lemma~\ref{lemSCPushforward}, 
we consider an open neighborhood $U \subseteq \Sigma$ of $K$ with compact closure $\overline{U}$
and we extend $f(\, \overline U\, )$ to a smooth spacelike Cauchy surface $\Sigma^\prime$ of $M^\prime$. 
With the same notation, we then find 
\begin{flalign}
\ip{(h,\tilde h)\, }{\, f_\ast (h',\tilde h'\, )}
= \ip{\iota_{\Sigma^\prime}^\ast (h,\tilde h)\, }{\,
  \iota_{\Sigma^\prime}^\ast\, (f^\ast)^{-1} (h',\tilde h'\, )}_{\Sigma'}
= \ip{\iota_{\Sigma^\prime}^\ast (h,\tilde h)\, }{\, (f_U^\ast)^{-1}
  \, j^\ast \, \iota_\Sigma^\ast (h',\tilde h'\, )}_{\Sigma'}~.
\end{flalign}
This follows from the definition of the pairing \eqref{eqPairingConf} and the diagram in \eqref{eqAux4}. 
Note that $f_{U \ast}: \dH^{k,m-k}_{\rm c}(U;\bbZ) \to \dH^{k,m-k}_{\rm c}(\Sigma^\prime;\bbZ)$ is defined as the colimit 
over the directed set $\mathcal K_U$ of compact subsets of $U$ of the inverse of the natural isomorphism 
$f_U^\ast: \dH^{k,m-k}(\Sigma^\prime,\Sigma^\prime \setminus f_U(-)) \Rightarrow \dH^{k,m-k}(U,U \setminus -)$, 
see \cite[Section~4]{BBSS}. 
We further use commutativity of the diagram
\begin{flalign}\label{eqAux}
\xymatrix{
\dH^k(\Sigma^\prime;\bbZ) \times \dH^{m-k}_{\rm c}(U;\bbZ)
	\ar[rr]^-{f_U^\ast \times \id} \ar[d]_-{\id \times f_{U \ast}}
		&&	\dH^k(U;\bbZ) \times \dH^{m-k}_{\rm c}(U;\bbZ) \ar[d]^-{\ips{\cdot}{\cdot}_{\rm c}}					\\
\dH^k(\Sigma^\prime;\bbZ) \times \dH^{m-k}_{\rm c}(\Sigma^\prime;\bbZ) \ar[rr]_-{\ips{\cdot}{\cdot}_{\rm c}}
		&&	\bbT
}
\end{flalign}
which is shown in \cite[Section~5.2]{BBSS}.
It then follows that 
\begin{flalign}
\ip{(h,\tilde h)\,}{\, f_\ast (h',\tilde h'\,)}
= \ip{\iota_{\Sigma^\prime}^\ast (h,\tilde h)\, }{\, f_{U \ast} \, j^\ast\, \iota_\Sigma^\ast (h',\tilde h'\, )}_{\Sigma'}
= \ip{f_U^\ast\, \iota_{\Sigma^\prime}^\ast (h,\tilde h)\, }{\, j^\ast \, \iota_\Sigma^\ast (h',\tilde h'\, )}_{U}~.
\end{flalign}
From diagram \eqref{eqAux3} we have $\iota_{\Sigma^\prime} \circ f_U = f \circ \iota_\Sigma \circ j$,
and by using the analogue of the diagram \eqref{eqAux} for the open embedding $j: U \to \Sigma$ we find
\begin{flalign}
\ip{(h,\tilde h)\, }{\, f_\ast (h',\tilde h'\, )}
= \ip{j^\ast\, \iota_\Sigma^\ast \, f^\ast(h,\tilde h)\, }{\, j^\ast\, \iota_\Sigma^\ast(h',\tilde h'\, )}_{U}
= \ip{\iota_\Sigma^\ast\, f^\ast(h,\tilde h)\, }{\, j_\ast\, j^\ast\, \iota_\Sigma^\ast(h',\tilde h'\, )}_{\Sigma}~.
\end{flalign}
Similarly to $f_{U \ast}$, the group homomorphism $j_\ast: \dH^{k,m-k}_{\rm c}(U;\bbZ) \to \dH^{k,m-k}_{\rm c}(\Sigma;\bbZ)$ 
is obtained as the colimit over $\mathcal K_U$ of the inverse of the natural isomorphism 
$j^\ast: \dH^{k,m-k}(\Sigma,\Sigma \setminus -;\bbZ) \Rightarrow \dH^{k,m-k}(U,U \setminus -;\bbZ)$. It follows that
\begin{flalign}
\ip{(h,\tilde h)\, }{\, f_\ast (h',\tilde h'\, )}
= \ip{\iota_\Sigma^\ast \, f^\ast(h,\tilde h)\, }{\, (j^\ast)^{-1}\, j^\ast\,  \iota_\Sigma^\ast(h',\tilde h'\,)}_{\Sigma}
=\ip{f^\ast (h,\tilde h)\, }{\, (h',\tilde h'\, )}~,
\end{flalign}
where for the last equality we used the definition of the pairing \eqref{eqPairingConf}. 
\end{proof}

\begin{lem}\label{lemIproperties}
Recalling \eqref{eqn:Imapping}, 
the group homomorphisms $I : \Conf^k_{\rm sc}(M;\bbZ)\to \Conf^k(M;\bbZ)$, for objects $M$ in $\Loc_m$,
enjoy the following properties:
\begin{itemize}
\item[(i)] $f^\ast\circ I \circ f_\ast = I$, for all morphisms $f : M\to M^\prime$ in $\Loc_m$.
\item[(ii)] $f^\ast_2\circ I \circ {f_{1}}_\ast =0$, for all diagrams $M_1\stackrel{f_1}{\longrightarrow} M 
\stackrel{f_2}{\longleftarrow} M_2$ in $\Loc_m$ such that the images of $f_1$ and $f_2$ are causally disjoint.
\end{itemize}
\end{lem}
\begin{proof}
Let us start with statement (i). For $(h',\tilde h'\, ) \in \Conf^k_{\rm sc}(M;\bbZ)$, 
we show that $f^\ast\, I\, f_\ast (h',\tilde h'\, ) = I (h',\tilde h'\, )$. 
For a fixed smooth spacelike Cauchy surface $\Sigma$, let $K \subseteq \Sigma$ be compact 
with $(h',\tilde h'\, ) \in \Conf^k(M,M \setminus J(K);\bbZ)$. 
By \cite[eq.\ (3.13)]{BBSS} one has 
$f^\ast\circ I = I\circ f^\ast: \Conf^k(M^\prime,M^\prime \setminus J(f(K));\bbZ) \to \Conf^k(M;\bbZ)$, and so we find 
\begin{flalign}
f^\ast\, I \, (f^\ast)^{-1} (h',\tilde h'\, ) = I \, f^\ast \, (f^\ast)^{-1} (h',\tilde h'\, ) = I (h',\tilde h'\, )~.
\end{flalign}
This equation corresponds to $f^\ast\, I \, f_\ast (h',\tilde h'\, ) = I (h',\tilde h'\, )$ when $(h',\tilde h'\, )$ 
is regarded as an element of $\Conf^k_{\rm sc}(M;\bbZ)$. 
\sk

For statement (ii), let $(h',\tilde h'\, ) \in \Conf^k_{\rm sc}(M_1;\bbZ)$. 
Choosing a smooth spacelike Cauchy surface $\Sigma$ of $M_1$, we find a compact subset $K$ of $\Sigma$ 
such that $(h',\tilde h'\, ) \in \Conf^k(M_1,M_1 \setminus J(K);\bbZ)$. 
As above, we represent $f_{1 \ast} (h',\tilde h'\, ) \in \Conf^k_{\rm sc}(M;\bbZ)$ 
by $(f_1^\ast)^{-1} (h',\tilde h'\, ) \in \Conf^k(M,M \setminus J(f_1(K));\bbZ)$. 
For each pair of cycles $(z,\tilde z) \in Z_{k-1,m-k-1}(M_2)$, 
the pushforward $f_{2 \ast} (z,\tilde z) \in Z_{k-1,m-k-1}(M)$ is supported inside $f_2(M_2)$. 
By assumption $f_2(M_2) \subseteq M \setminus J(f_1(M_1)) \subseteq M \setminus J(f_1(K))$, 
hence $f_{2 \ast} (z,\tilde z) = 0 $ in $ Z_{k-1,m-k-1}(M,M \setminus J(f_1(K)))$. 
In particular, $I \, (f_1^\ast)^{-1} (h',\tilde h'\, )$ vanishes when evaluated on $f_{2 \ast} (z,\tilde z)$. 
Since this is the case for any $(z,\tilde z) \in Z_{k-1,m-k-1}(M_2)$, 
we conclude that $f_2^\ast\, I \, (f_1^\ast)^{-1} \, (h',\tilde h'\, ) = 0$ and 
hence also $f_2^\ast\, I \, f_{1 \ast} (h',\tilde h'\, ) = 0$. 
\end{proof}

%%%%%%%%%%%%%%%%%%%%%%%%%%%%%%%%%%%%%%%%%%%%%%%%
%%%%%%%%%%%%%%%%%%%%%%%%%%%%%%%%%%%%%%%%%%%%%%%%

\end{document}